\documentclass[11pt]{article}
\usepackage{amsthm,amstext,amssymb,amsmath,graphicx,amsfonts,caption,subcaption,mathrsfs,multirow}
\usepackage[utf8]{inputenc}
\usepackage{booktabs}
\usepackage{tikz}
    \usetikzlibrary{arrows,topaths}
    \usepackage{placeins}
\usepackage[mathscr]{euscript}
\usepackage{multicol}
\usepackage{natbib}
\usepackage[pagebackref=false,colorlinks,linkcolor=blue,citecolor=blue]{hyperref}

\theoremstyle{plain}
\newtheorem{theorem}{Theorem}[section]
\newtheorem{lemma}[theorem]{Lemma}
\newtheorem{proposition}[theorem]{Proposition}

\theoremstyle{definition}
\newtheorem{definition}[theorem]{Definition}
\newtheorem{corollary}[theorem]{Corollary}
\newtheorem{example}[theorem]{Example}

\theoremstyle{remark}
\newtheorem{remark}{\sc Remark}
\makeatother
\makeatletter
\def\namedlabel#1#2{\begingroup
   \def\@currentlabel{#2}%
   \label{#1}\endgroup}
\makeatother
\date{}
\title{\bf The pure spectrum of a residuated lattice}\vspace{.25 in}
\author{ \vspace{.25 in} {\bf Saeed Rasouli$^{1*}$} and {\bf Amin Dehghani $^2$}\\
Department of Mathematics, Persian Gulf University, \\Bushehr, Iran \\
{\tt $^1$srasouli@pgu.ac.ir }\\
{\tt $^2$dehghany.amin@hotmail.com}\\}
 
\begin{document}
 \maketitle
 \begin{abstract}
This paper studies a fascinating type of filter in residuated lattices, the so-called pure filters. A combination of algebraic and topological methods on the pure filters of a residuated lattice is applied to obtain some new and structural results. The notion of purely-prime filters of a residuated lattice has been investigated, and a Cohen-type theorem has been obtained. It is shown that the pure spectrum of a residuated lattice is a compact sober space, and a Grothendieck-type theorem has been demonstrated. It is proved that the pure spectrum of a Gelfand residuated lattice is a Hausdorff space, and deduced that the pure spectrum of a Gelfand residuated lattice is homeomorphic to its usual maximal spectrum. Finally, the pure spectrum of an mp-residuated lattice is investigated and verified that a given residuated lattice is mp iff its minimal prime spectrum is equipped with the induced dual hull-kernel topology, and its pure spectrum is the same.\footnote{
$^*$: Corresponding auther}
\end{abstract}
%%%%%%%%%%%%%%%%%%%%%%%%%%%%%%%%%%%%%%%%%%%%%%%%%%%%%%%%%%%%%%%%%%%%%%%%%%%%%%%%%%%%%%%%%
\section{Introduction}

This paper deals with the pure filters of a residuated lattice as an essential tool in studying residuated lattices. Historically, this notion is rooted in the idea of \textit{pure subgroups}, introduced by \cite{prufer1923untersuchungen} for describing conditions of the decomposition of primary abelian groups as direct sums of cyclic groups. A subgroup $H$ of an abelian group $G$ is said to be \textit{pure} provided that $H\cap nG=nH$, for all integers $n$. The work of Pr\"{u}fer complemented by \cite{kulikoff1941theorie}, where many results proved again using pure subgroups systematically. A complete discussion of pure subgroups, their relation to infinite abelian group theory, and a survey of their literature are given in \cite{kaplansky1954infinite}. Pure subgroups were generalized in several ways in the theory of abelian groups and modules. Pure submodules were defined in a variety of ways but eventually settled on the modern definition in terms of tensor products or systems of equations. Let $R$ be a ring and $M$ an $R$-module. A submodule $N$ of $M$ is called a \textit{pure submodule} by \cite{cohn1959free} provided that the sequence $0\rightarrow N\otimes E\rightarrow M\otimes E$ is exact, for all $R$-module $E$. \cite{anderson1974rings} called the submodule $N$ a pure submodule of $M$ provided that $N\cap IM=IN$, for all ideal $I$ of $R$. \cite{ribenboim1972algebraic} defined $N$ to be pure in $M$ if $N\cap aM=aN$, for all $a\in R$. Although the first condition implies the second and the second implies the third, these definitions are not equivalent in general (see, e.g. \citet[\S 4J]{lam1999lectures}). Following \citet[Theorem 4.85 \& Corollary 4.86]{lam1999lectures}, if $M/N$ is flat, then $N$ is pure in $M$, and when $M$ is flat the converse holds. Following \citet[Corollaries 11.21 \& 11.23]{faith1973algebra}, if $M$ is flat, then for any
submodule $N$ of $M$ the following are equivalent:
\begin{itemize}
  \item $M/N$ is flat;
  \item $N\cap IM=IN$, for all ideal $I$ of $R$;
  \item $N\cap aM=aN$, for all $a\in R$.
\end{itemize}

Thus, for example, the definitions coincide for purity of ideals of $R$. The notion of pure ideals has been extensively studied during past decades, see e.g., \cite{bkouche1970purete}, \cite{fieldhouse1974purity}, \cite{de1983projectivity}, \cite{borceux1983algebra}, \cite{borceux1984sheaf}, \cite{al1988pure}, \cite{al1989pure}, \cite{tarizadeh2019flat}, \cite{tarizadeh2021purely}, \cite{tarizadeh2021some}. Also, pure ideals are investigated under the name of \textit{neat ideals} in \citet[p. 188]{johnstone1982stone}. Inspired by the purity notion, many authors have proposed similar notions, under other names, for various structures over the years, see e.g., $\sigma$-ideals in bounded distributive lattices \citep{cornish1972normal,cornish1977ideals,al1991sigma}, pure ideals in MV-algebras \citep{georgescu1997pure}, pure filters in BL-algebras \citep{leustean2003representations}, pure i-filters in residuated lattices \citep{bucsneag2012stable}, pure congruences in universal algebra \cite{georgescu2022reticulation}. Pure ideals in lattices are also the lattice-theoretic version of \textit{virginal ideals} in rings were introduced by \cite{borceux1984sheaf} for obtaining a sheaf representation applicable to all rings and their modules. In the spirit of \cite{borceux1984sheaf}, the notion of virginal elements in a multiplicative ideal structure (mi-structure) is introduced and investigated by \cite{georoescu1989some}.

It is known that residuated lattices play an unshakable role in the theory of fuzzy logic. Lots of logical algebras such as MTL-algebras, divisible residuated lattices, BL-algebras, MV-algebras, Heyting algebras, and Boolean algebras are subvarieties of residuated lattices. Hence, the notion of purity in residuated lattices is indeed a generalization of the notion of purity in these structures.

Given the above discussion, we decided to take a deeper and more extensive look at pure filters in residuated lattices. Therefore, the notion of purity is investigated, and as an outcome, various interesting algebraic and topological results are discovered. Although, this concept in residuated lattices has been reviewed by \cite{rasouli2021rickart}, however, here we present some additional properties for this concept that seem to shed more light on algebraic and topological situations. In this paper, we study the purely-prime filters of a residuated lattice, and obtain various new and engaging results. The notion of purely-prime filters, like prime filters, is fascinating. This paper tries to show interesting aspects of
this topic and some of its spectacular applications. Our findings show that some of the results obtained by some of the above papers can also be reproduced through residuated lattices.

This paper is organized into six sections as follows: In Sect.\ref{sec2}, some definitions and facts about residuated lattices are recalled, and some of their propositions are extracted. We illustrate this section by some examples of residuated lattices, which will be used in the following sections. Sect. \ref{sec3} deals with the notion of pure filters of a residuated lattice as an essential tool in studying residuated lattices, and some algebraic characterizations are given for them. It has been pointed out that the pure filters in a residuated lattice are generalizations of direct summands and shown that the set of pure filters forms a frame.We show that in a residuated lattice every filter is pure if and only if it is hyperarchimedean and give a characterization for a hyperarchimedean residuated lattice using the $\mathscr{D}$-topology. This section is closed by giving some properties of the pure part of a filter of a residuated lattice. Sect. \ref{sec4} investigates the notion of purely-prime and purely-maximal filters of a residuated lattice. Theorem \ref{t1spaspp} gives a necessary and sufficient condition for the pure spectrum of a residuated lattice to be a $T_1$ space. It has been shown that the pure spectrum of a residuated lattice is a compact sober space. Theorem \ref{spsppconti} verifies that The pure part map of a residuated lattice is continuous. We demonstrate that $Spp$ is a contravariant functor from the category of residuated lattices to the category of topological spaces. A Grothendieck-type theorem has been obtained, which states that there is a canonical correspondence between the complemented elements of a residuated lattice and the clopens
of its pure spectrum. In the end, we establish that each purely-maximal filter of a residuated lattice is principal if and only if each pure filter of it is principal. In Sect. \ref{sec5} the notion of purity is used to obtain some new characterizations for Gelfand residuated lattices. The Gelfand residuated lattices have been characterized using the sink of a filter and described by the pure part of a filter. Theorem \ref{equgelchapure} represented the relationship between the pure filters and radicals in a Gelfand residuated lattice. Theorem \ref{gelfmaxpure} describes the purely-maximal filters of a Gelfand residuated lattice, and shows that the pure spectrum of a Gelfand residuated lattice is an antichain. Proposition \ref{gelspphau} results in that the pure spectrum of a Gelfand residuated lattice is a Hausdorff space, and Theorem \ref{sppgelfch} deduces that the pure spectrum of a Gelfand residuated lattice is homeomorphic to its usual maximal spectrum. Finally, Theorem \ref{gelpurefcl} characterizes precisely the pure filters of a Gelfand residuated lattice, and Theorem \ref{gelfhulldmin} gives a characterization for a Gelfand residuated lattice using the $\mathscr{D}$-topology. In Sect. \ref{sec6} the notion of purity is used to obtain some new characterizations for mp-residuated lattices. Theorem \ref{norgammsig} gives some characterizations for mp-residuated lattices in terms of purity, and Theorem \ref{mppurefcl} characterizes precisely the pure filters of an mp-residuated lattice. The main result of this section is Theorem \ref{mp2minspp} which establishes that a residuated lattice is mp if and only if the set of its minimal prime filters and the set of its purely-prime filters coincide. Theorem \ref{equmpflatmin} verifies that a residuated lattice is mp if and only if its pure spectrum and minimal prime spectrum, equipped with the dual hull-kernel topology, are homeomorphic with the identity map. Corollary \ref{mpspphau} results in that the pure spectrum of an mp-residuated lattice is a Hausdorff space. This paper is closed by Theorem \ref{minspprick} which implies that an mp-residuated lattice is Rickart if and only if its pure spectrum and usual minimal prime spectrum are homeomorphic.

\section{Residuated lattices}\label{sec2}

In this section, we recall some definitions, properties and results relative to residuated lattices which will be used
in the following.

An algebra $\mathfrak{A}=(A;\vee,\wedge,\odot,\rightarrow,0,1)$ is called a \textit{residuated lattice} provided that $\ell(\mathfrak{A})=(A;\vee,\wedge,0,1)$ is a bounded lattice, $(A;\odot,1)$ is a commutative monoid, and $(\odot,\rightarrow)$ is an adjoint pair. A residuated lattice $\mathfrak{A}$ is called \textit{non-degenerate} if $0\neq 1$. For a residuated lattice $\mathfrak{A}$, and $a\in A$ we put $\neg a:=a\rightarrow 0$, and $a^n:=a\odot\cdots\odot a$ ($n$ times), for every integer $n$. The class of residuated lattices is equational, and so forms a variety. For a survey of residuated lattices, the interested reader is referred to \cite{galatos2007residuated}.
\begin{remark}\label{resproposition}\citep[Proposition 2.6]{ciungu2006classes}
Let $\mathfrak{A}$ be a residuated lattice. The following conditions are satisfied for any $x,y,z\in A$:
\begin{enumerate}
  \item [$(r_{1})$ \namedlabel{res1}{$(r_{1})$}] $x\odot (y\vee z)=(x\odot y)\vee (x\odot z)$;
  \item [$(r_{2})$ \namedlabel{res2}{$(r_{2})$}] $x\vee (y\odot z)\geq (x\vee y)\odot (x\vee z)$.
  \end{enumerate}
\end{remark}
\begin{example}\label{exa6}
Let $A_6=\{0,a,b,c,d,1\}$ be a lattice whose Hasse diagram is given by Figure \ref{figa6}. Routine calculation shows that $\mathfrak{A}_6=(A_6;\vee,\wedge,\odot,\rightarrow,0,1)$ is a residuated lattice in which the commutative operation $``\odot"$ is given by Table \ref{taba6}, and the operation $``\rightarrow"$ is given by $x\rightarrow y=\bigvee \{z\in A_6|x\odot z\leq y\}$, for any $x,y\in A_6$.
\FloatBarrier
\begin{table}[ht]
\begin{minipage}[b]{0.56\linewidth}
\centering
\begin{tabular}{ccccccc}
\hline
$\odot$ & 0 & a & b & c & d & 1 \\ \hline
0       & 0 & 0 & 0 & 0 & 0 & 0 \\
        & a & a & a & 0 & a & a \\
        &   & b & a & 0 & a & b \\
        &   &   & c & c & c & c \\
        &   &   &   & d & d & d \\
        &   &   &   &   & 1 & 1 \\ \hline
\end{tabular}
\caption{Cayley table for ``$\odot$" of $\mathfrak{A}_6$}
\label{taba6}
\end{minipage}\hfill
\begin{minipage}[b]{0.6\linewidth}
\centering
  \begin{tikzpicture}[>=stealth',semithick,auto]
    \tikzstyle{subj} = [circle, minimum width=6pt, fill, inner sep=0pt]
    \tikzstyle{obj}  = [circle, minimum width=6pt, draw, inner sep=0pt]

    \tikzstyle{every label}=[font=\bfseries]

    % Before diagram .........................
    \node[subj,  label=below:0] (0) at (0,0) {};
    \node[subj,  label=below:c] (c) at (-1,1) {};
    \node[subj,  label=below:a] (a) at (1,.5) {};
    \node[subj,  label=below right:b] (b) at (1,1.5) {};
    \node[subj,  label=below:d] (d) at (0,2) {};
    \node[subj,  label=below right:1] (1) at (0,3) {};

    \path[-]   (0)    edge                node{}      (a);
    \path[-]   (a)    edge                node{}      (b);
    \path[-]   (0)    edge                node{}      (c);
    \path[-]   (c)    edge                node{}      (d);
    \path[-]   (b)    edge                node{}      (d);
    \path[-]   (d)    edge                node{}      (1);
\end{tikzpicture}
\captionof{figure}{Hasse diagram of $\mathfrak{A}_{6}$}
\label{figa6}
\end{minipage}
\end{table}
\FloatBarrier
\end{example}
\begin{example}\label{exb6}
Let $B_6=\{0,a,b,c,d,1\}$ be a lattice whose Hasse diagram is given by Figure \ref{figb6}. Routine calculation shows that $\mathfrak{B}_6=(B_6;\vee,\wedge,\odot,\rightarrow,0,1)$ is a residuated lattice in which the commutative operation $``\odot"$ is given by Table \ref{tabb6}, and the operation $``\rightarrow"$ is given by $x\rightarrow y=\bigvee \{z\in B_6|x\odot z\leq y\}$, for any $x,y\in B_6$.
\FloatBarrier
\begin{table}[ht]
\begin{minipage}[b]{0.56\linewidth}
\centering
\begin{tabular}{ccccccc}
\hline
$\odot$ & 0 & a & b & c & d & 1 \\ \hline
0       & 0 & 0 & 0 & 0 & 0 & 0 \\
        & a & a & 0 & a & 0 & a \\
        &   & b & 0 & 0 & b & b \\
        &   &   & c & a & b & c \\
        &   &   &   & d & d & d \\
        &   &   &   &   & 1 & 1 \\ \hline
\end{tabular}
\caption{Cayley table for ``$\odot$" of $\mathfrak{B}_6$}
\label{tabb6}
\end{minipage}\hfill
\begin{minipage}[b]{0.6\linewidth}
\centering
  \begin{tikzpicture}[>=stealth',semithick,auto]
    \tikzstyle{subj} = [circle, minimum width=6pt, fill, inner sep=0pt]
    \tikzstyle{obj}  = [circle, minimum width=6pt, draw, inner sep=0pt]

    \tikzstyle{every label}=[font=\bfseries]

    % Before diagram .........................
    \node[subj,  label=below:0] (0) at (0,0) {};
    \node[subj,  label=below:a] (a) at (-1,1) {};
    \node[subj,  label=below:b] (b) at (1,1) {};
    \node[subj,  label=below:c] (c) at (0,2) {};
    \node[subj,  label=below:d] (d) at (2,2) {};
    \node[subj,  label=below:1] (1) at (1,3) {};

    \path[-]   (0)    edge                node{}      (a);
    \path[-]   (0)    edge                node{}      (b);
    \path[-]   (b)    edge                node{}      (d);
    \path[-]   (d)    edge                node{}      (1);
    \path[-]   (a)    edge                node{}      (c);
    \path[-]   (b)    edge                node{}      (c);
    \path[-]   (c)    edge                node{}      (1);
\end{tikzpicture}
\captionof{figure}{Hasse diagram of $\mathfrak{B}_{6}$}
\label{figb6}
\end{minipage}
\end{table}
\FloatBarrier
\end{example}
\begin{example}\label{exc6}
  Let $C_6=\{0,a,b,c,d,1\}$ be a lattice whose Hasse diagram is given by Figure \ref{figc6}. Routine calculation shows that $\mathfrak{C}_6=(C_6;\vee,\wedge,\odot,\rightarrow,0,1)$ is a residuated lattice in which the commutative operation $``\odot"$ is given by Table \ref{tabc6}, and the operation $``\rightarrow"$ is given by $x\rightarrow y=\bigvee \{z\in C_6|x\odot z\leq y\}$, for any $x,y\in C_6$.
\FloatBarrier
\begin{table}[ht]
\begin{minipage}[b]{0.56\linewidth}
\centering
\begin{tabular}{ccccccc}
\hline
$\odot$ & 0 & a & b & c & d & 1 \\ \hline
0       & 0 & 0 & 0 & 0 & 0 & 0 \\
        & a & 0 & 0 & 0 & 0 & a \\
        &   & b & 0 & 0 & 0 & b \\
        &   &   & c & 0 & 0 & c \\
        &   &   &   & d & 0 & d \\
        &   &   &   &   & 1 & 1 \\ \hline
\end{tabular}
\caption{Cayley table for ``$\odot$" of $\mathfrak{C}_6$}
\label{tabc6}
\end{minipage}\hfill
\begin{minipage}[b]{0.6\linewidth}
\centering
  \begin{tikzpicture}[>=stealth',semithick,auto]
    \tikzstyle{subj} = [circle, minimum width=6pt, fill, inner sep=0pt]
    \tikzstyle{obj}  = [circle, minimum width=6pt, draw, inner sep=0pt]

    \tikzstyle{every label}=[font=\bfseries]

    % Before diagram .........................
    \node[subj,  label=below:0] (0) at (0,0) {};
    \node[subj,  label=below:c] (c) at (-1,1) {};
    \node[subj,  label=below:a] (a) at (1,.5) {};
    \node[subj,  label=below right:b] (b) at (1,1.5) {};
    \node[subj,  label=below:d] (d) at (0,2) {};
    \node[subj,  label=below right:1] (1) at (0,3) {};

    \path[-]   (0)    edge                node{}      (a);
    \path[-]   (a)    edge                node{}      (b);
    \path[-]   (0)    edge                node{}      (c);
    \path[-]   (c)    edge                node{}      (d);
    \path[-]   (b)    edge                node{}      (d);
    \path[-]   (d)    edge                node{}      (1);
\end{tikzpicture}
\captionof{figure}{Hasse diagram of $\mathfrak{C}_{6}$}
\label{figc6}
\end{minipage}
\end{table}
\FloatBarrier
\end{example}
\begin{example}\label{exa8}
Let $A_8=\{0,a,b,c,d,e,f,1\}$ be a lattice whose Hasse diagram is given by Figure \ref{figa8}. Routine calculation shows that $\mathfrak{A}_8=(A_8;\vee,\wedge,\odot,\rightarrow,0,1)$ is a residuated lattice in which the commutative operation $``\odot"$ is given by Table \ref{taba8}, and the operation $``\rightarrow"$ is given by $x\rightarrow y=\bigvee \{z\in A_8|x\odot z\leq y\}$, for any $x,y\in A_8$.
\FloatBarrier
\begin{table}[ht]
\begin{minipage}[b]{0.56\linewidth}
\centering
\begin{tabular}{ccccccccc}
\hline
$\odot$ & 0 & a & b & c & d & e & f & 1 \\ \hline
0       & 0 & 0 & 0 & 0 & 0 & 0 & 0 & 0 \\
        & a & 0 & a & a & a & a & a & a \\
        &   & b & 0 & 0 & 0 & 0 & b & b \\
        &   &   & c & c & a & c & a & c \\
        &   &   &   & d & a & a & d & d \\
        &   &   &   &   & e & c & d & e \\
        &   &   &   &   &   & f & f & f \\
        &   &   &   &   &   &   & 1 & 1 \\ \hline
\end{tabular}
\caption{Cayley table for ``$\odot$" of $\mathfrak{A}_8$}
\label{taba8}
\end{minipage}\hfill
\begin{minipage}[b]{0.6\linewidth}
\centering
  \begin{tikzpicture}[>=stealth',semithick,auto]
    \tikzstyle{subj} = [circle, minimum width=6pt, fill, inner sep=0pt]
    \tikzstyle{obj}  = [circle, minimum width=6pt, draw, inner sep=0pt]

    \tikzstyle{every label}=[font=\bfseries]

    % Before diagram .........................
    \node[subj,  label=below:0] (0) at (0,0) {};
    \node[subj,  label=below:a] (a) at (-1,1) {};
    \node[subj,  label=below:b] (b) at (1,1) {};
    \node[subj,  label=below:c] (c) at (-2,2) {};
    \node[subj,  label=below:d] (d) at (0,2) {};
    \node[subj,  label=below:e] (e) at (-1,3) {};
    \node[subj,  label=below:f] (f) at (1,3) {};
    \node[subj,  label=below:1] (1) at (0,4) {};

    \path[-]   (0)    edge                node{}      (a);
    \path[-]   (0)    edge                node{}      (b);
    \path[-]   (b)    edge                node{}      (d);
    \path[-]   (d)    edge                node{}      (f);
    \path[-]   (f)    edge                node{}      (1);
    \path[-]   (a)    edge                node{}      (d);
    \path[-]   (a)    edge                node{}      (c);
    \path[-]   (c)    edge                node{}      (e);
    \path[-]   (d)    edge                node{}      (e);
    \path[-]   (e)    edge                node{}      (1);
\end{tikzpicture}
\captionof{figure}{Hasse diagram of $\mathfrak{A}_{8}$}
\label{figa8}
\end{minipage}
\end{table}
\FloatBarrier
\end{example}

Let $\mathfrak{A}$ be a residuated lattice. A non-void subset $F$ of $A$ is called a \textit{filter} of $\mathfrak{A}$ provided that $x,y\in F$ implies $x\odot y\in F$, and $x\vee y\in F$, for any $x\in F$ and $y\in A$. The set of filters of $\mathfrak{A}$ is denoted by $\mathscr{F}(\mathfrak{A})$. A filter $F$ of $\mathfrak{A}$ is called \textit{proper} if $F\neq A$. Clearly, $F$ is proper iff $0\notin F$. For any subset $X$ of $A$, the \textit{filter of $\mathfrak{A}$ generated by $X$} is denoted by $\mathscr{F}(X)$. For each $x\in A$, the filter generated by $\{x\}$ is denoted by $\mathscr{F}(x)$ and said to be \textit{principal}. The set of principal filters is denoted by $\mathscr{PF}(\mathfrak{A})$. Following \citet[\S 5.7]{gratzer2011lattice}, a join-complete lattice $\mathfrak{A}$, is called a \textit{frame} if it satisfies the join infinite distributive law (JID), i.e., for any $a\in A$ and $S\subseteq A$, $a\wedge \bigvee S=\bigvee \{a\wedge s\mid s\in S\}$.  According to \cite{galatos2007residuated}, $(\mathscr{F}(\mathfrak{A});\cap,\veebar,\textbf{1},A)$ is a frame in which $\veebar \mathcal{F}=\mathscr{F}(\cup \mathcal{F})$, for any $\mathcal{F}\subseteq \mathscr{F}(\mathfrak{A})$. With any filter $F$ of a residuated lattice $\mathfrak{A}$, we can associate a binary relation $\equiv_{F}$ on $A$ as follows; $x\equiv_{F} y$ if and only if $x\rightarrow y,y\rightarrow x\in F$, for any $x,y\in A$. The binary relation $\equiv_{F}$ is a congruence on $\mathfrak{A}$, and called \textit{the congruence induced by $F$ on $\mathfrak{A}$}. As usual, the set of all congruences on $\mathfrak{A}$ shall be denoted by $Con(\mathfrak{A})$. It is well-known that $\phi:\mathscr{F}(\mathfrak{A})\longrightarrow Con(\mathfrak{A})$, defined by $\phi(F)=\equiv_F$, is a lattice isomorphism, and this implies that $\mathcal{RL}$ is an ideal determined variety. For a filter $F$ of a residuated lattice $\mathfrak{A}$, the quotient set $A/\equiv_{F}$ with the natural operations becomes a residuated lattice which is denoted by $\mathfrak{A}/F$. Also, for any $a\in A$, the equivalence classes $a/\equiv_{F}$ is denoted by $a/F$.

\begin{example}\label{filterexa}
Consider the residuated lattice $\mathfrak{A}_6$ from Example \ref{exa6}, the residuated lattice $\mathfrak{C}_6$ from Example \ref{exb6}, the residuated lattice $\mathfrak{C}_6$ from Example \ref{exc6}, and the residuated lattice $\mathfrak{A}_8$ from Example \ref{exa8}. The sets of their filters are presented in Table \ref{tafiex}.
\begin{table}[h]
\centering
\begin{tabular}{ccl}
\hline
                 & \multicolumn{2}{c}{Filters}                                       \\ \hline
$\mathfrak{A}_6$ & \multicolumn{2}{c}{$\{1\},\{a,b,d,1\},\{c,d,1\},\{d,1\},A_6$} \\
$\mathfrak{B}_6$ & \multicolumn{2}{c}{$\{1\},\{a,c,1\},\{d,1\},B_6$} \\
$\mathfrak{C}_6$ & \multicolumn{2}{c}{$\{1\},C_6$} \\
$\mathfrak{A}_8$ & \multicolumn{2}{c}{$\{1\},\{a,c,d,e,f,1\},\{c,e,1\},\{f,1\},A_8$} \\ \hline
\end{tabular}
\caption{The sets of filters of $\mathfrak{A}_6$, $\mathfrak{B}_6$, $\mathfrak{C}_{6}$, and $\mathfrak{A}_8$}
\label{tafiex}
\end{table}
\end{example}

The following proposition has a routine verification.
\begin{proposition}\label{genfilprop}
Let $\mathfrak{A}$ be a residuated lattice and $F$ be a filter of $\mathfrak{A}$. The following assertions hold, for any $x,y\in A$:
\begin{enumerate}
 \item  [(1) \namedlabel{genfilprop1}{(1)}] $\mathscr{F}(F,x)\stackrel{def.}{=}F\veebar \mathscr{F}(x)=\{a\in A|f\odot x^n\leq a,\textrm{~for~some}~f\in F~\textrm{and}~n\in \mathbb{N}\}$;
  \item  [(2) \namedlabel{genfilprop2}{(2)}] $x\leq y$ implies $\mathscr{F}(F,y)\subseteq \mathscr{F}(F,x)$.
  \item  [(3) \namedlabel{genfilprop3}{(3)}] $\mathscr{F}(F,x)\cap \mathscr{F}(F,y)=\mathscr{F}(F,x\vee y)$;
  \item  [(4) \namedlabel{genfilprop4}{(4)}] $\mathscr{F}(F,x)\veebar \mathscr{F}(F,y)=\mathscr{F}(F,x\odot y)=\mathscr{F}(F\cup \{x,y\})$;
  \item  [(5) \namedlabel{genfilprop5}{(5)}] $\mathscr{PF}(\mathfrak{A})$ is a sublattice of $\mathscr{F}(\mathfrak{A})$.
  %\item  [(6) \namedlabel{genfilprop6}{(6)}] if $\mathfrak{A}$ is finite, then $\mathscr{F}(\mathfrak{A})=\mathscr{PF}(\mathfrak{A})$.
\end{enumerate}
\end{proposition}

Let $\mathfrak{A}$ and $\mathfrak{B}$ be residuated lattices. A map $f:A\longrightarrow B$ is called a  \textit{morphism}, in symbols $f:\mathfrak{A}\longrightarrow \mathfrak{B}$, if it preserves the fundamental operations. If $f:\mathfrak{A}\longrightarrow \mathfrak{B}$ is a morphism we put $coker(f)=f^{\leftarrow}(1)$. One can see that $coker(f)$ is a filter of $\mathfrak{A}$. Also, it is obvious that $f$ is a monomorphism if and only if $coker(f)=\{1\}$.
\begin{proposition}\cite[Proposition 2.5]{rasouli2019going}\label{prorh}
Let $f:\mathfrak{A}\longrightarrow \mathfrak{B}$ be a residuated morphism. The following assertions hold:
\begin{enumerate}
  \item [(1) \namedlabel{prorh1}{(1)}] $F\in \mathscr{F}(\mathfrak{B})$ implies $f^\leftarrow(F)\in \mathscr{F}(\mathfrak{A})$ and $coker(f)\subseteq f^\leftarrow(F)$;
  \item [(2) \namedlabel{prorh2}{(2)}] if $f$ is onto and $F$ is a filter of $\mathfrak{A}$ containing $coker(f)$, then $f(F)$ is a filter of $\mathfrak{B}$.
\end{enumerate}
\end{proposition}
\begin{proposition}\cite[Corollary 2.8]{rasouli2019going}\label{filqou}
Let $f:\mathfrak{A}\longrightarrow \mathfrak{B}$ be an epimorphism. Then we have
\[\mathscr{F}(\mathfrak{B})=\{f(F)\mid coker(f)\subseteq F\in \mathscr{F}(\mathfrak{A})\}.\]
\end{proposition}

A proper filter of a residuated lattice $\mathfrak{A}$ is said to be \textit{maximal} provided that it is a maximal element in the set of all proper filters of $\mathfrak{A}$. The set of maximal filters of $\mathfrak{A}$ shall be denoted by $Max(\mathfrak{A})$. A proper filter $\mathfrak{p}$ of $\mathfrak{A}$ is said to be \textit{prime} provided that $x\vee y\in \mathfrak{p}$ implies $x\in \mathfrak{p}$ or $y\in \mathfrak{p}$, for any $x,y\in A$. The set of prime filters of $\mathfrak{A}$ shall be denoted by $Spec(\mathfrak{A})$. Since $\mathscr{F}(\mathfrak{A})$ is a distributive lattice, so $Max(\mathfrak{A})\subseteq Spec(\mathfrak{A})$. Zorn's lemma verifies that any proper filter is contained in a maximal filter, and so in a prime filter.
%\begin{theorem}\cite[Theorem 3.18]{rasouli2019going}\label{prfilth}
%If $\mathscr{C}$ is a $\vee$-closed subset of $\mathfrak{A}$ which does not meet the filter $F$, then $F$ is contained in a filter $P$ which is maximal with respect to the property of not meeting $\mathscr{C}$; furthermore $P$ is prime.
%\end{theorem}

% In the following, for a subset $X$ of $A$ and a collection of prime filters $\Pi$ of $\mathfrak{A}$, we set $_{X}\Pi=\{P\in \Pi\mid X\subseteq P\}$ and $^{X}\Pi=\{P\in \Pi\mid P\subseteq X\}$.
\begin{proposition}\cite[Corollary 3.19]{rasouli2019going}\label{intprimfilt}
Let $F$ be a filter of a residuated lattice $\mathfrak{A}$ and $X$ a subset of $A$. The following assertions hold:
\begin{enumerate}
\item  [(1) \namedlabel{intprimfilt1}{(1)}]  If $X\nsubseteq F$, there exists a prime filter $P$ containing $F$ such that $X\nsubseteq P$;
\item  [(2) \namedlabel{intprimfilt2}{(2)}] $\mathscr{F}(X)=\bigcap \{\mathfrak{p}\mid X\subseteq \mathfrak{p}\in Spec(\mathfrak{A})\}$.
\end{enumerate}
\end{proposition}

Let $\mathfrak{A}$ be a residuated lattice. If $X$ is a subset of $A$, a prime filter $P$ of $\mathfrak{A}$ is called a \textit{minimal prime filter belonging to $X$} provided that $P$ is a minimal element in the set of prime filters containing $X$. The set of minimal prime filters belonging to $X$ is denoted by $Min_{X}(\mathfrak{A})$. A minimal prime filter belonging to $\{1\}$ is simply called a minimal prime filter, and the set of minimal prime filters of $\mathfrak{A}$ is denoted by $Min(\mathfrak{A})$. For the basic facts concerning prime filters of a residuated lattice,  the reader is referred to \cite{rasouli2019going}.
\begin{example}\label{maxminex}
Consider the residuated lattice $\mathfrak{A}_6$ from Example \ref{exa6}, the residuated lattice $\mathfrak{B}_6$ from Example \ref{exb6}, the residuated lattice $\mathfrak{C}_6$ from Example \ref{exc6}, and the residuated lattice $\mathfrak{A}_8$ from Example \ref{exa8}. The sets of their maximal, prime, and minimal prime filters are presented in Table \ref{prfiltab}.
\begin{table}[h]
\centering
\begin{tabular}{cccc}
\hline
                 & \multicolumn{3}{c}{Prime filters}      \\ \hline
                 & Maximal filters &          & Minimal prime filters      \\
$\mathfrak{A}_6$ &  $\{a,b,d,1\},\{c,d,1\}$                 &  &$ \{1\}$\\
$\mathfrak{B}_6$ &  $\{a,c,1\},\{d,1\}$                 &  &$\{a,c,1\},\{d,1\}$\\
$\mathfrak{C}_6$ &  $\{1\}$                                     &         & $\{1\}$\\
$\mathfrak{A}_8$ &  $\{a,c,d,e,f,1\}$                          &         & $\{c,e,1\},\{f,1\}$ \\ \hline
\end{tabular}
\caption{The sets of maximal, prime, and minimal prime filters of $\mathfrak{A}_6$, $\mathfrak{B}_6$, $\mathfrak{C}_{6}$, and $\mathfrak{A}_8$}
\label{prfiltab}
\end{table}
\end{example}
%\begin{proposition}\label{1mineq}
%Let $\mathfrak{A}$ be a residuated lattice. A subset $P$ of $A$ is a minimal prime filter iff $P^c\stackrel{def.}{=}A\setminus P$ is a $\vee$-closed subset of $\mathfrak{A}$ which it is maximal with respect to the property of not containing $1$.
%\end{proposition}
% \begin{proof}
%   It follows by \citet[Theorem 3.24]{rasouli2019going}
% \end{proof}
\begin{proposition}\cite[Corollary 3.25]{rasouli2019going}\label{mp}
   Any prime filter of a residuated lattice contains a minimal prime filter.
 \end{proposition}

Let $(A;\vee,\wedge,0,1)$ be a bounded lattice. Recall \citep[\S I.6.1]{gratzer2011lattice} that an element $x\in A$ is
called \textit{complemented} if there is an element $y\in A$ such that $x\vee y=1$ and $x\wedge y=0$; $y$ is called a complement of $x$. Complements in a bounded lattice are generally not unique unless the lattice is distributive. If $y$ is the unique complement of $x$, we write $y=\ddot{x}$. If $\mathfrak{A}$ is a residuated lattice, the set of complemented elements of $\ell(\mathfrak{A})$ is denoted by $\beta(\mathfrak{A})$ and called \textit{the Boolean center} of $\mathfrak{A}$.  In residuated lattices, however, although the underlying lattices need not be distributive, the complements are unique. In the following, for a residuated lattice $\mathfrak{A}$ we set $\mathscr{F}(\beta(\mathfrak{A}))=\{\mathscr{F}(e)\mid e\in \beta(\mathfrak{A})\}$. For further study about the Boolean center of a residuated lattice we refer to \cite{georgescu2014Boolean}.
\begin{example}\label{compeleex}
Consider the residuated lattice $\mathfrak{A}_6$ from Example \ref{exa6}, the residuated lattice $\mathfrak{B}_6$ from Example \ref{exb6}, the residuated lattice $\mathfrak{C}_6$ from Example \ref{exc6}, and the residuated lattice $\mathfrak{A}_8$ from Example \ref{exa8}. Then $\beta(\mathfrak{A}_6)=\{0,1\}$, $\beta(\mathfrak{B}_6)=\{0,a,d,1\}$, $\beta(\mathfrak{C}_6)=\{0,1\}$ and $\beta(\mathfrak{A}_8)=\{0,1\}$
\end{example}
\begin{proposition}\cite[Lemma 2.5]{georgescu2014Boolean}\label{boleleprop}
Let $\mathfrak{A}$ be a residuated lattice. The following assertions hold:
\begin{enumerate}
  \item  [$(1)$ \namedlabel{boleleprop1}{$(1)$}] $\mathscr{F}(e)=\{a\in A|e\leq a\}$, for any $e\in \beta(\mathfrak{A})$;
  \item  [$(2)$ \namedlabel{boleleprop2}{$(2)$}] $\beta(\mathfrak{A})=\{a\in A\mid a\vee \neg a=1\}$;
  \item  [$(3)$ \namedlabel{boleleprop3}{$(3)$}] $\ddot{e}=\neg e$, for any $e\in \beta(\mathfrak{A})$;
%  \item  [$(4)$ \namedlabel{boleleprop4}{$(4)$}] $e\rightarrow x=\neg e\vee x$, for any $e\in \beta(\mathfrak{A})$ and $x\in A$;
%  \item  [$(5)$ \namedlabel{boleleprop5}{$(5)$}] $\mathscr{F}(\textbf{e})=(\mathscr{F}(e);\vee,\wedge,\odot,\rightarrow,e,1)$ is a residuated lattice;
%  \item  [$(3)$ \namedlabel{boleleprop6}{$(3)$}] $e\leq x\vee \neg x$, implies that $e\odot x\in \beta(\mathfrak{A})$, for any $e\in \beta(\mathfrak{A})$ and $x\in A$;
  \item  [$(4)$ \namedlabel{boleleprop7}{$(4)$}] $e\odot x=e\wedge x$, for any $e\in \beta(\mathfrak{A})$ and $x\in A$.
\end{enumerate}
\end{proposition}

According to \citet[Definition 4.37]{mckenzie2018algebras}, a residuated lattice is said to be \textit{directly indecomposable} if and only if it is not isomorphic to a direct product of two non-degenerate residuated lattices.
\begin{proposition}\cite[Proposition 3.4]{rasouli2021rickart}\label{direcindbeta}
  A residuated lattice $\mathfrak{A}$ is directly indecomposable if and only if $\beta(\mathfrak{A})=\{0,1\}$.
\end{proposition}

Let $\mathfrak{A}$ be a residuated lattice and $F$ a filter of $\mathfrak{A}$. For any subset $X$ of $A$, we write $(F:X)=\{a\in A|a\vee x\in F,\forall x\in X\}$, and we set:
 \begin{itemize}
   \item $\Gamma_{F}(\mathfrak{A})=\{(F:X)|X\subseteq A\}$;
   \item $\gamma_{F}(\mathfrak{A})=\{(F:x)|x\in A\}$.
 \end{itemize}

Elements of $\Gamma_{F}(\mathfrak{A})$ and $\gamma_{F}(\mathfrak{A})$ are called \textit{$F$-coannihilators} and \textit{$F$-coannulets}. By \citet[Proposition 3.13]{rasouli2018generalized} follows that $(\Gamma_{F}(\mathfrak{A});\cap,\vee^{\Gamma_{F}},F,A)$ is a complete Boolean lattice, in which for any $\mathscr{F}\subseteq \Gamma_{F}(\mathfrak{A})$ we have $\vee^{\Gamma_{F}} \mathscr{F}=(F:(F:\cup \mathscr{F}))$. By \citet[Proposition 2.15]{rasouli2020quasicomplemented} follows that $\gamma_{F}(\mathfrak{A})$ is a sublattice of $\Gamma_{F}(\mathfrak{A})$. In the sequel, if $F=\{1\}$, $(F:X)$ is denoted by $X^{\perp}$, and the subscript $F$ is dropped. For the basic facts concerning coannihilators and coannulets of residuated lattices we refer to \cite{rasouli2018generalized}.

Let $\mathfrak{A}$ be a residuated lattice. A complemented element of the lattice $\mathscr{F}(\mathfrak{A})$ is called \textit{a direct summand} of $\mathfrak{A}$.
\begin{proposition}\cite[Proposition 4.2]{rasouli2021rickart}\label{b9fxpro}
Let $\mathfrak{A}$ be a residuated lattice and $F$ a filter of $\mathfrak{A}$. The following assertions are equivalent:
\begin{enumerate}
\item  [$(1)$ \namedlabel{b9fxpro1}{$(1)$}] $F$ is a direct summand of $\mathfrak{A}$;
\item  [$(2)$ \namedlabel{b9fxpro2}{$(2)$}] $F\veebar F^{\perp}=A$;
\item  [$(3)$ \namedlabel{b9fxpro3}{$(3)$}] $F\in \mathscr{F}(\beta(\mathfrak{A}))$.
\end{enumerate}
\end{proposition}

A residuated lattice $\mathfrak{A}$ is called \textit{hyperarchimedean} provided that all its elements are archimedean. By \citet[Theorem 50]{bucsneag2010archimedean} follows that a residuated lattice $\mathfrak{A}$ is hyperarchimedean if and only if it is zero-dimensional, i.e. $(Spec(\mathfrak{A});\subseteq)$ is an antichain. By \textit{Nachbin theorem} \citep[Theorem 123]{gratzer2011lattice} follows that each Boolean lattice is a hyperarchimedean residuated lattice. The following proposition derives some necessary and sufficient conditions for a residuated lattice to become hyperarchimedean.
\begin{proposition}\cite[Proposition 3.6]{rasouli2021rickart}\label{hperarchpri}
 Let $\mathfrak{A}$ be a residuated lattice. The following assertions are equivalent:
\begin{enumerate}
\item  [$(1)$ \namedlabel{hperarchpri1}{$(1)$}] $\beta(\mathscr{F}(\mathfrak{A}))=\mathscr{PF}(\mathfrak{A})$;
\item  [$(2)$ \namedlabel{hperarchpri2}{$(2)$}] $\mathscr{F}(\beta(\mathfrak{A}))=\mathscr{PF}(\mathfrak{A})$;
\item  [$(3)$ \namedlabel{hperarchpri3}{$(3)$}] $\mathfrak{A}$ is hyperarchimedean.
\end{enumerate}
\end{proposition}
\begin{proposition}\citep[Theorem 8.6]{rasouli2019going}\label{1mineq}
Let $\mathfrak{A}$ be a residuated lattice. A prime filter $\mathfrak{p}$ of $\mathfrak{A}$ is minimal prime if and only if  for any $x\in A$, $\mathfrak{p}$ contains precisely one of $x$ or $x^{\perp}$.
\end{proposition}

%Let $\mathfrak{A}_1$ and $\mathfrak{A}_2$ be residuated lattices. Recalling that a mapping $h:A_1\longrightarrow A_2$ is called a  \textit{residuated morphism}, in symbols $h:\mathfrak{A}_1\longrightarrow \mathfrak{A}_2$, if it preserves the fundamental operations. It is evident that $coker(h)=h^{\leftarrow}(1)$ is a filter of $\mathfrak{A}_1$.
%\begin{proposition}\label{homeb}
%  Let $\mathfrak{A}$ be a residuated lattice and $e\in \beta(\mathfrak{A})$. Then $\mathfrak{A}/\mathscr{F}(e)\cong \mathscr{F}(\neg \textbf{e})$.
%\end{proposition}
%\begin{proof}
%  It is obvious that $h_e:\mathfrak{A}\longrightarrow \mathscr{F}(\textbf{e})$, defined by $h_e(a)=a\vee e$, is an epimorphism and $coker(h_e)=e^{\perp}$. So the result holds by Proposition \ref{annperp}.
%\end{proof}

A morphism of rings $f:\mathfrak{A}\longrightarrow \mathfrak{B}$ is said to be \textit{flat} provided that $\mathfrak{B}$ is flat as an $\mathfrak{A}$-module. According to \citet[p. 46]{picavet1980ultrafiltres}, a morphism of commutative rings $f:\mathfrak{A}\longrightarrow \mathfrak{B}$ is flat if and only if for any ideal $I$ of $\mathfrak{A}$, and for any finitely generated ideal $J$ of $\mathfrak{A}$, we have $(I:J)B=(IB:JB)$. By a suggestion of \citet[p. 8]{georgescu2020}, this characterization of flat ring morphisms in terms of residuation allows us to define a notion of ``flateness" for the residuated morphisms.
\begin{definition}
  Let $f:\mathfrak{A}\longrightarrow \mathfrak{B}$ be a residuated morphism. We say that $f$ fulfills \textit{the flatness property} ($f$ is \textit{flat}) provided that for any filter $F$ of $\mathfrak{A}$ and $a\in A$, $f$ satisfies the following property:
  \[\mathscr{F}(f((F:a)))=(\mathscr{F}(f(F)):f(a)).\]
\end{definition}

The following lemma has a routine verification, and so its proof is left to the reader.
\begin{lemma}\label{flatmorprop}
  Let $f:\mathfrak{A}\longrightarrow \mathfrak{B}$ be a residuated morphism. The following assertions are equivalent:
  \begin{enumerate}
\item  [$(1)$ \namedlabel{flatmorprop1}{$(1)$}] $f$ is flat;
\item  [$(2)$ \namedlabel{flatmorprop2}{$(2)$}] $(\mathscr{F}(f(F)):f(a))\subseteq \mathscr{F}(f((F:a)))$, for any filter $F$ of $\mathfrak{A}$ and $a\in A$.
\end{enumerate}
\end{lemma}

Let $\mathfrak{A}$ be a residuated lattice and $F$ a filter of $\mathfrak{A}$. It is well-known that the canonical projection $\pi_{F}:\mathfrak{A}\longrightarrow\mathfrak{A}/F$, given by $\pi_{F}(x)=x/F$, is an epimorphism and $coker(\pi_{F})=F$.
\begin{proposition}\label{canonflat}
  Let $\mathfrak{A}$ be a residuated lattice and $F$ a filter of $\mathfrak{A}$. The following assertions are equivalent:
  \begin{enumerate}
\item  [$(1)$ \namedlabel{canonflat1}{$(1)$}] $\pi_{F}$ is flat;
\item  [$(2)$ \namedlabel{canonflat2}{$(2)$}] $(G\veebar F:a)\subseteq (G:a)\veebar F$, for any filter $G$ of $\mathfrak{A}$ and $a\in A$.
\end{enumerate}
\end{proposition}
\begin{proof}
It proves routinely, because $\mathscr{F}(\pi_{F}(G))=G\veebar F/F$, for any filter $G$ of $\mathfrak{A}$.
\end{proof}

One of a type of important filters in residuated lattices are $\alpha$-filters. A filter $F$ of a residuated lattice $\mathfrak{A}$ is called an $\alpha$-\textit{filter} provided that for any $x\in F$ we have $x^{\perp\perp}\subseteq F$. The set of $\alpha$-filters of $\mathfrak{A}$ is denoted by $\alpha(\mathfrak{A})$. For any subset $X$ of $A$, the $\alpha$-filter generated by $X$ is denoted by $\alpha(X)$. By \citet[Proposition 5.3]{rasouli2020quasicomplemented} follows that $(\alpha(\mathfrak{A});\cap,\vee^{\alpha},\{1\},A)$ is a frame, in which $\vee^{\alpha} \mathcal{F}=\alpha(\veebar \mathcal{F})$, for any $\mathcal{F}\subseteq \alpha(\mathfrak{A})$ . For the basic facts concerning $\alpha$-filters of residuated lattices we refer to \cite{rasouli2020quasicomplemented}.
\begin{example}\label{alpfilterexa}
Consider the residuated lattice $\mathfrak{A}_6$ from Example \ref{exa6}, the residuated lattice $\mathfrak{B}_6$ from Example \ref{exb6}, the residuated lattice $\mathfrak{C}_6$ from Example \ref{exc6}, and the residuated lattice $\mathfrak{A}_8$ from Example \ref{exa8}. The sets of their $\alpha$-filters are presented in Table \ref{alphatafiex}.
\begin{table}[h]
\centering
\begin{tabular}{ccl}
\hline
                 & \multicolumn{2}{c}{$\alpha$-filters}                                       \\ \hline
$\mathfrak{A}_6$ & \multicolumn{2}{c}{$\{1\},A_6$} \\
$\mathfrak{B}_6$ & \multicolumn{2}{c}{$\{1\},\{a,c,1\},\{d,1\},B_6$} \\
$\mathfrak{C}_6$ & \multicolumn{2}{c}{$\{1\},C_6$} \\
$\mathfrak{A}_8$ & \multicolumn{2}{c}{$\{1\},\{c,e,1\},\{f,1\},A_8$} \\ \hline
\end{tabular}
\caption{The sets of $\alpha$-filters of  $\mathfrak{A}_6$, $\mathfrak{B}_6$, $\mathfrak{C}_{6}$, and $\mathfrak{A}_8$}
\label{alphatafiex}
\end{table}
\end{example}

 Let $\mathfrak{A}$ be a residuated lattice. For an ideal $I$ of $\ell(\mathfrak{A})$, set $\omega(I)=\{a\in A|a\vee x=1,\textrm{~for~some~} x\in I\}$, and $\Omega(\mathfrak{A})=\{\omega(I)|I\in id(\ell(\mathfrak{A}))\}$. Using Proposition 3.4 of \cite{rasouli2018n}, it follows that $\Omega(\mathfrak{A})\subseteq \mathscr{F}(\mathfrak{A})$, and so elements of $\Omega(\mathfrak{A})$ are called \textit{$\omega$-filters} of $\mathfrak{A}$. For an $\omega$-filter $F$ of $\mathfrak{A}$, $I_{F}$ denoted an ideal of $\ell(\mathfrak{A})$, which satisfies $F=\omega(I_{F})$. \citet[Proposition 3.7]{rasouli2018n} shows that $(\Omega(\mathfrak{A});\cap,\vee^{\omega},\{1\},A)$ is a bounded distributive lattice, in which $F\vee^{\omega} G=\omega(I_{F}\curlyvee I_{G})$, for any $F,G\in \Omega(\mathfrak{A})$ (by $\curlyvee$, we mean the join operation in the lattice of ideals of $\ell(\mathfrak{A})$). For any proper filter $H$ of $\mathfrak{A}$ we set $D(H)=\omega(\dot{H})$. For the basic facts concerning $\omega$-filters of a residuated lattice, interested readers are referred to \cite{rasouli2018n}.
\begin{proposition}\citep{rasouli2018n}\label{omegprop}
    Let $\mathfrak{A}$ be residuated lattice. The following assertions hold:
    \begin{enumerate}
\item  [$(1)$ \namedlabel{omegprop1}{$(1)$}] $\gamma(\mathfrak{A})$ is a sublattice of $\Omega(\mathfrak{A})$;%  \cite[Proposition 3.9]{rasouli2018n};
\item  [$(2)$ \namedlabel{omegprop2}{$(2)$}]  $D(\mathfrak{p})=k\mathscr{G}(\mathfrak{p})=k(\mathscr{G}(\mathfrak{p})\cap Min(\mathfrak{A}))$, for any prime filter $\mathfrak{p}$ of $\mathfrak{A}$;% \cite[Corollary 3.23]{rasouli2018n};
\item  [$(3)$ \namedlabel{omegprop3}{$(3)$}] a prime filter $\mathfrak{p}$ of $\mathfrak{A}$ is minimal prime if and only if $\mathfrak{p}=D(\mathfrak{p})$.
%\item  [$(4)$ \namedlabel{omegprop4}{$(4)$}] $\Omega(\mathfrak{A})\subseteq \alpha(\mathfrak{A})$; %\cite[Proposition 5.11]{rasouli2020quasicomplemented};
%\item  [$(5)$ \namedlabel{omegprop5}{$(5)$}] $\omega(X)=A$ if and only if $1\in X$.% \cite[Proposition 3.2(6)]{rasouli2018n};
%\item  [$(6)$ \namedlabel{omegprop6}{$(6)$}] $D(P)=
%\item  [$(7)$ \namedlabel{omegprop7}{$(7)$}]
\end{enumerate}
  \end{proposition}

Let $\mathfrak{A}$ be a residuated lattice and $\Pi$ a collection of prime filters of $\mathfrak{A}$.
For a subset $\pi$ of $\Pi$ we set $k(\pi)=\bigcap \pi$, and for a subset $X$ of $A$ we set $h_{\Pi}(X)=\{P\in \Pi\mid X\subseteq  P\}$ and $d_{\Pi}(X)=\Pi\setminus h_{\Pi}(X)$. The collection $\Pi$ can be topologized by taking the collection $\{h_{\Pi}(x)\mid x\in A\}$  as a closed (an open) basis, which is called \textit{the (dual) hull-kernel topology} on $\Pi$ and denoted by $\Pi_{h(d)}$. The generated
topology by $\tau_{h}\cup \tau_{d}$ on $Spec(\mathfrak{A})$ is called \textit{the patch topology} and denoted by $\tau_{p}$. In the sequel, subscripts $M$ and $m$ are reserved for $Max(\mathfrak{A})$ and $Min(\mathfrak{A})$, respectively. For a detailed discussion on the (dual) hull-kernel and patch topologies on a residuated lattice, we refer to \cite{rasouli2018hull}.
\begin{proposition}\citep{rasouli2018hull}\label{hulkerinstr}
Let $\mathfrak{A}$ be a residuated lattice. $Max_{h}(\mathfrak{A})$ and $Min_{d}(\mathfrak{A})$ are compact.
%The following assertions hold:
%\begin{enumerate}
%\item  [$(1)$ \namedlabel{hulkerinstr1}{$(1)$}] $Max_{h}(\mathfrak{A})$ and $Min_{d}(\mathfrak{A})$ are compact;
%\item  [$(2)$ \namedlabel{hulkerinstr2}{$(2)$}] $Spec_{h}(\mathfrak{A})$ and $Max_{h}(\mathfrak{A})$ are compact.
%\end{enumerate}
\end{proposition}
The following proposition characterizes the open sets of the prime spectrum of a residuated lattice w.r.t the dual hull-kernel topology.
\begin{proposition}\label{opensd}
  Let $\mathfrak{A}$ be a residuated lattice. The open sets of $Spec_{d}(\mathfrak{A})$ are precisely of the form $\{\mathfrak{p}\in Spec(\mathfrak{A})\mid \mathfrak{p}\cap X\neq \emptyset\}$, where $X$ is a subset of $A$.
\end{proposition}
\begin{proof}
  Let $U$ be an open set in $Spec_{d}(\mathfrak{A})$. So $U=\bigcup_{x\in X} h(x)$, for some $X\subseteq A$. It is clear that $\bigcup_{x\in X} h(x)=\{\mathfrak{p}\in Spec(\mathfrak{A})\mid \mathfrak{p}\cap X\neq \emptyset\}$. This holds the result.
\end{proof}
\begin{remark}\label{closd}
  Let $\mathfrak{A}$ be a residuated lattice. By Proposition \ref{opensd} follows that the closed sets of $Spec_{d}(\mathfrak{A})$ are precisely of the form $\{\mathfrak{p}\in Spec(\mathfrak{A})\mid \mathfrak{p}\cap X=\emptyset\}$, where $X$ is a subset of $A$.
\end{remark}

Let $\Pi$ be a collection of prime filters in a residuated lattice $\mathfrak{A}$. If $\pi$ is a subset of $\Pi$, its \textit{specialization} (\textit{generalization}) in $\Pi$, $\mathscr{S}_{\Pi}(\pi)$ ($\mathscr{G}_{\Pi}(\pi)$), is the set of all primes in $\Pi$, which contain (are contained in) some prime belonging to $\pi$; $\pi$ is said to be $\mathscr{S}_{\Pi}$-stable ($\mathscr{G}_{\Pi}$-stable) if and only if $\mathscr{S}_{\Pi}(\pi)=\pi$ ($\mathscr{G}_{\Pi}(\pi)=\pi$). If $\Pi$ is understood, it will be dropped. Notice that for any subset $B$ of $A$, $\bigcup_{b\in B}h(b)(\bigcup_{b\in B}d(b))$ is $\mathscr{S(G)}$-stable. The following theorem characterizes the closed sets of the hull-kernel topology.

%Let $\Pi$ be a collection of prime filters in a residuated lattice $\mathfrak{A}$. In the following, for a given subset $\pi$ of $\Pi$, $cl^{\Pi}_{h(d)}(\pi)$ stands for the closure of $\pi$ in the topological space $(\Pi,\tau_{h(d)})$. If $\pi=\{P\}$ for some prime filter $P$ of $\mathfrak{A}$, then $cl^{\Pi}_{h(d)}(\{P\})$ is simply denoted by $cl^{\Pi}_{h(d)}(P)$. If $\Pi$ is understood, it will be dropped.
%\begin{lemma}\citep[Theorem 3.14]{rasouli2018hull}\label{retractlemma}
%Let $\mathfrak{A}$ be a residuated lattice, $\Pi$ a collection of prime filters of $\mathfrak{A}$ and $\mathfrak{p},\mathfrak{q}\in \Pi$. The following assertions hold:
%\begin{enumerate}
%   \item [(1) \namedlabel{retractlemma1}{(1)}] $cl_{h}(\mathfrak{p})=\mathscr{S}(\mathfrak{p})$;
%   \item [(2) \namedlabel{retractlemma2}{(2)}] $cl_{d}(\mathfrak{p})=\mathscr{G}(\mathfrak{p})$.
% \end{enumerate}
%\end{lemma}

\begin{theorem}\cite[Theorem 4.30]{rasouli2018hull}\label{closefalzai}
  Let $\mathfrak{A}$ be a residuated lattice and $\pi$ a subset of $Spec(\mathfrak{A})$. $\pi$ is closed under the hull-kernel topology if and only if it is closed under the patch topology and $\mathscr{S}$-stable.
\end{theorem}

%%%%%%%%%%%%%%%%%%%%%%%%%%%%%%%%%%%%%%%%%%%%%%%%%%%%%%%%%%%%%%%%%%%%%%%%%%%%%%%%%%%%%%%%%%%%%%%%%%%%%%%%%%%%%%
\section{Pure filters}\label{sec3}

In this section, the notion of pure filters of a residuated lattice is investigated, and some properties of them are extracted.
\begin{definition}\label{sigfildef}
Let $\mathfrak{A}$ be a residuated lattice. For any filter $F$ of $\mathfrak{A}$ we set $\sigma(F)=k\mathscr{G}h(F)$, and shall be called the \textit{sink} of $F$ (see Fig. \ref{figsigma}).
\end{definition}
\begin{figure}[h]

\centering

\tikzset{every picture/.style={line width=0.75pt}} %set default line width to 0.75pt

\begin{tikzpicture}[x=0.75pt,y=0.75pt,yscale=-1,xscale=1]
%uncomment if require: \path (0,468); %set diagram left start at 0, and has height of 468

%Curve Lines [id:da5345081318221856]
\draw    (258.88,80.57) .. controls (288.8,49.88) and (370.87,50.55) .. (400.54,80.24) ;
%Straight Lines [id:da6237955980365555]
\draw    (258.88,80.57) -- (330.2,270) ;
%Straight Lines [id:da4277707910641657]
\draw    (400.54,80.24) -- (330.2,270) ;
%Straight Lines [id:da5667922958569298]
\draw    (258.88,80.57) -- (329.86,150.3) ;
%Straight Lines [id:da2944693229965707]
\draw    (400.54,80.24) -- (329.86,150.3) ;
%Flowchart: Connector [id:dp5441110512213791]
\draw  [fill={rgb, 255:red, 0; green, 0; blue, 0 }  ,fill opacity=1 ] (328.33,70.4) .. controls (328.33,69.38) and (329.11,68.56) .. (330.06,68.56) .. controls (331.02,68.56) and (331.8,69.38) .. (331.8,70.4) .. controls (331.8,71.41) and (331.02,72.23) .. (330.06,72.23) .. controls (329.11,72.23) and (328.33,71.41) .. (328.33,70.4) -- cycle ;
%Straight Lines [id:da10068343297195015]
\draw    (330.06,70.23) -- (360.75,99.75) ;
%Straight Lines [id:da5071303934668836]
\draw    (330.06,70.4) -- (300.75,99.35) ;
%Straight Lines [id:da4323782877276323]
\draw    (300.75,99.35) -- (330.35,191.35) ;
%Straight Lines [id:da5559721445740466]
\draw    (360.75,99.75) -- (330.35,191.35) ;
%Flowchart: Connector [id:dp9859298600384907]
\draw  [fill={rgb, 255:red, 0; green, 0; blue, 0 }  ,fill opacity=1 ] (328.13,150.3) .. controls (328.13,149.28) and (328.91,148.46) .. (329.86,148.46) .. controls (330.82,148.46) and (331.59,149.28) .. (331.59,150.3) .. controls (331.59,151.31) and (330.82,152.14) .. (329.86,152.14) .. controls (328.91,152.14) and (328.13,151.31) .. (328.13,150.3) -- cycle ;
%Flowchart: Connector [id:dp6153188867981925]
\draw  [fill={rgb, 255:red, 0; green, 0; blue, 0 }  ,fill opacity=1 ] (328.35,190.35) .. controls (328.35,189.33) and (329.13,188.51) .. (330.08,188.51) .. controls (331.04,188.51) and (331.81,189.33) .. (331.81,190.35) .. controls (331.81,191.36) and (331.04,192.19) .. (330.08,192.19) .. controls (329.13,192.19) and (328.35,191.36) .. (328.35,190.35) -- cycle ;
%Flowchart: Connector [id:dp883178057133891]
\draw  [fill={rgb, 255:red, 0; green, 0; blue, 0 }  ,fill opacity=1 ] (328.73,270) .. controls (328.73,268.98) and (329.51,268.16) .. (330.46,268.16) .. controls (331.42,268.16) and (332.2,268.98) .. (332.2,270) .. controls (332.2,271.02) and (331.42,271.84) .. (330.46,271.84) .. controls (329.51,271.84) and (328.73,271.02) .. (328.73,270) -- cycle ;
%Straight Lines [id:da5217854948750247]
\draw    (330.4,57.2) -- (490.85,57.6) ;
%Straight Lines [id:da2019365225976466]
\draw    (330.46,270) -- (490.91,270.4) ;
%Shape: Brace [id:dp8821787051702679]
\draw   (499.5,270) .. controls (504.17,269.98) and (506.49,267.64) .. (506.47,262.97) -- (506.04,173.08) .. controls (506.01,166.41) and (508.33,163.07) .. (513,163.04) .. controls (508.33,163.07) and (505.98,159.75) .. (505.95,153.08)(505.96,156.08) -- (505.54,65.97) .. controls (505.52,61.3) and (503.18,58.98) .. (498.51,59) ;
%Straight Lines [id:da6603288955313542]
\draw  [dash pattern={on 4.5pt off 4.5pt}]  (200.18,57.17) -- (240.5,57.18) -- (330.4,57.2) ;
%Straight Lines [id:da054101979059473226]
\draw  [dash pattern={on 4.5pt off 4.5pt}]  (199.64,148.43) -- (239.96,148.44) -- (329.86,148.46) ;
%Shape: Brace [id:dp489968878093775]
\draw   (189.5,58) .. controls (184.83,58.05) and (182.53,60.41) .. (182.58,65.08) -- (182.89,93.08) .. controls (182.96,99.75) and (180.67,103.11) .. (176,103.16) .. controls (180.67,103.11) and (183.04,106.41) .. (183.11,113.08)(183.08,110.08) -- (183.43,141.08) .. controls (183.48,145.75) and (185.84,148.05) .. (190.51,148) ;

% Text Node
\draw (331,160) node    {$F$};
% Text Node
\draw (341,70) node    {$\mathfrak{p}$};
% Text Node
\draw (332,201) node    {$D(\mathfrak{p})$};
% Text Node
\draw (331,284) node    {$\sigma(F)$};
% Text Node
\draw (540,163) node    {$\mathscr{G}h(F)$};
% Text Node
\draw (143,103) node    {$h(F)$};

\end{tikzpicture}
\caption{} \label{figsigma}
\end{figure}

The following proposition has a routine verification, and so its proof is left to the reader.
\begin{proposition}\label{sigmapro}
  Let $\mathfrak{A}$ be a residuated lattice. The following assertions hold:
\begin{enumerate}
\item  [$(1)$ \namedlabel{sigmapro1}{$(1)$}] $\sigma(F)$ is a filter of $\mathfrak{A}$, for any $F\in \mathscr{F}(\mathfrak{A})$;
\item  [$(2)$ \namedlabel{sigmapro2}{$(2)$}] $F\subseteq G$ implies $\sigma(F)\subseteq \sigma(G)$, for any $F,G\in \mathscr{F}(\mathfrak{A})$;
\item  [$(3)$ \namedlabel{sigmapro3}{$(3)$}] $\sigma(\mathfrak{p})\subseteq D(\mathfrak{p})$, for any $\mathfrak{p}\in Spec(\mathfrak{A})$;
\item  [$(4)$ \namedlabel{sigmapro4}{$(4)$}] $\sigma(\mathfrak{m})=D(\mathfrak{m})$, for any $\mathfrak{m}\in Max(\mathfrak{A})$;
\end{enumerate}
\end{proposition}
%\begin{proof}
%  \item [\ref{sigmapro1}:] It is evident.
%  \item [\ref{sigmapro2}:] It follows by $\Sigma(G)\subseteq \Sigma(F)$.
%  \item [\ref{sigmapro3}:] It follows by $^{P}Spec(\mathfrak{A})\subseteq \Sigma(P)$.
%  \item [\ref{sigmapro4}:] It follows by $^{M}Spec(\mathfrak{A})=\Sigma(M)$.
%  \end{proof}
  \begin{theorem}\label{sigmafequiv}
  Let $\mathfrak{A}$ be a residuated lattice and $F$ a filter of $\mathfrak{A}$. The following assertions hold:
  \begin{enumerate}
\item  [$(1)$ \namedlabel{sigmafequiv1}{$(1)$}] $\sigma(F)=k(\mathscr{G}h(F)\cap Min(\mathfrak{A}))$;
\item  [$(2)$ \namedlabel{sigmafequiv2}{$(2)$}] $~~~~~~=\bigcap \{D(\mathfrak{p})\mid \mathfrak{p}\in h(F)\}$;
\item  [$(3)$ \namedlabel{sigmafequiv3}{$(3)$}] $~~~~~~=\{a\in A\mid a^{\perp}\veebar F=A\}$;
\item  [$(4)$ \namedlabel{sigmafequiv4}{$(4)$}] $~~~~~~=\{a\in A\mid \neg b\in F,\textrm{~for~some}~b\in a^{\perp}\}$;
\item  [$(5)$ \namedlabel{sigmafequiv5}{$(5)$}] $~~~~~~=\bigcap \{D(\mathfrak{m})\mid \mathfrak{m}\in h_{M}(F)\}$;
\item  [$(6)$ \namedlabel{sigmafequiv6}{$(6)$}] $~~~~~~=\omega(I_F)$, where $I_F=\{a\in A\mid a^{\perp\perp}\veebar F=A\}$.
\end{enumerate}
\end{theorem}
\begin{proof}
  \ref{sigmafequiv1} and \ref{sigmafequiv2} are obvious.
  \item [\ref{sigmafequiv3}:] Let $a\in \sigma(F)$. By absurdum assume that $a^{\perp}\veebar F\neq A$. So $a^{\perp}\subseteq \mathfrak{p}$, for some $\mathfrak{p}\in h(F)$. This yields that $a\notin D(\mathfrak{p})$; a contradiction. Conversely, let $a^{\perp}\veebar F=A$, for some $a\in A$, and $\mathfrak{p}\in h(F)$. So $b\odot f=0$, for some $b\in a^{\perp}$ and $f\in F$. This implies that $\neg b\in \mathfrak{p}$. Thus $b\notin \mathfrak{p}$, and so $a\in D(\mathfrak{p})$. This holds the result.
  \item [\ref{sigmafequiv4}:] There is nothing to prove.
  \item [\ref{sigmafequiv5}:] There is nothing to prove.
  \item [\ref{sigmafequiv6}:] Let $a\in \omega(I_F)$. So $a\in x^{\perp}$, for some $x\in I_F$. This implies that $A=x^{\perp\perp}\veebar F\subseteq a^{\perp}\veebar F$ and so $a\in \sigma(F)$. Conversely, let $a\in \sigma(F)$. So $x\odot f=0$, for some $x\in a^{\perp}$ and $f\in F$. This shows that $x\in I_F$ and so $a\in \omega(I_F)$.
\end{proof}
\begin{proposition}\label{primesigmad}
  Let $\mathfrak{A}$ be a residuated lattice. The following assertions hold:
\begin{enumerate}
\item  [$(1)$ \namedlabel{primesigmad1}{$(1)$}] $\sigma(F)\subseteq F$, for every $F\in \mathscr{F}(\mathfrak{A})$;
\item  [$(2)$ \namedlabel{primesigmad2}{$(2)$}] $\sigma(F\cap G)=\sigma(F)\cap \sigma(G)$, for any $F,G\in \mathscr{F}(\mathfrak{A})$;
\item  [$(3)$ \namedlabel{primesigmad3}{$(3)$}] $\veebar_{F\in \mathcal{F}}\sigma(F)\subseteq \sigma(\veebar \mathcal{F})$, for any $\mathcal{F}\subseteq \mathscr{F}(\mathfrak{A})$.
\end{enumerate}
\end{proposition}
\begin{proof}
  \item [\ref{primesigmad1}:] Let $F$ be a filter of $\mathfrak{A}$ and $a\in \sigma(F)$. Then for any $\mathfrak{p}\in h(F)$ we have $a\in D(\mathfrak{p})\subseteq \mathfrak{p}$. This states that $a\in kh(F)=F$.
  \item [\ref{primesigmad2}:] Let $F$ and $G$ be two filters of $\mathfrak{A}$ and $a\in \sigma(F)\cap \sigma(G)$. Since $a^{\perp}\veebar (F\cap G)=(a^{\perp}\veebar F)\cap (a^{\perp}\veebar G)=A$, so $a\in \sigma(F\cap G)$. The inverse is evident by Proposition \ref{sigmapro}\ref{sigmapro2}.
  \item [\ref{primesigmad3}:] It is evident by Proposition \ref{sigmapro}\ref{sigmapro2}.
\end{proof}

%\begin{definition}\cite[Definition 4.16]{rasouli2018n}
%Let $\mathfrak{A}$ be a residuated lattice. For any filter $F$ of $\mathfrak{A}$ we set
%\[\sigma(F)=\{a\in A|a^{\perp}\veebar F=A\}.\]
%\end{definition}
%\begin{proposition}\cite[Proposition 4.17]{rasouli2018n}\label{sigmaomegafil}
%Let $\mathfrak{A}$ be a residuated lattice and $F$ be a filter of $\mathfrak{A}$. Then $\sigma(F)$ is an $\omega$-filter contained in $F$.
%\end{proposition}

\begin{definition}\label{puredef}
  Let $\mathfrak{A}$ be a residuated lattice. A filter $F$ of $\mathfrak{A}$ is called \textit{pure} provided that $\sigma(F)=F$. The set of pure filters of $\mathfrak{A}$ is denoted by $\sigma(\mathfrak{A})$. It is obvious that $\{1\},A\in \sigma(\mathfrak{A})$.
\end{definition}
\begin{example}\label{purefilterexa}
Consider the residuated lattice $\mathfrak{A}_6$ from Example \ref{exa6}, the residuated lattice $\mathfrak{B}_6$ from Example \ref{exb6}, the residuated lattice $\mathfrak{C}_6$ from Example \ref{exc6}, and the residuated lattice $\mathfrak{A}_8$ from Example \ref{exa8}. The sets of their pure filters are presented in Table \ref{tafiex}.
\begin{table}[h]
\centering
\begin{tabular}{ccl}
\hline
                 & \multicolumn{2}{c}{Pure filters}                                       \\ \hline
$\mathfrak{A}_6$ & \multicolumn{2}{c}{$\{1\},A_6$} \\
$\mathfrak{B}_6$ & \multicolumn{2}{c}{$\{1\},\{a,c,1\},\{d,1\},B_6$} \\
$\mathfrak{C}_6$ & \multicolumn{2}{c}{$\{1\},C_6$} \\
$\mathfrak{A}_8$ & \multicolumn{2}{c}{$\{1\},A_8$} \\ \hline
\end{tabular}
\caption{The sets of pure filters of $\mathfrak{A}_6$, $\mathfrak{B}_6$, $\mathfrak{C}_{6}$, and $\mathfrak{A}_8$}
\label{tafiex}
\end{table}
\end{example}

The following theorem states that the pure filters in a residuated lattice are generalizations of direct summands.
\begin{theorem}\label{fbetasig}
  Let $\mathfrak{A}$ be a residuated lattice. Then any direct summand of $\mathfrak{A}$ is a pure filter of $\mathfrak{A}$.
\end{theorem}
\begin{proof}
Let $F$ be a direct summand of $\mathfrak{A}$. By Proposition \ref{b9fxpro} follows that $F=\mathscr{F}(e)$, for some $e\in \beta(\mathfrak{A})$. Let $a\in F$. Thus $A=e^{\perp}\veebar \mathscr{F}(e)\subseteq a^{\perp}\veebar F$.
\end{proof}

The following theorem characterizes pure filters of a residuated lattice in terms of flat residuated morphisms, inspired by the one obtained for unitary rings by \citet[Proposition 2]{borceux1983algebra}.
\begin{theorem}\label{flatpurethe}
  Let $\mathfrak{A}$ be a residuated lattice and $F$ a filter of $\mathfrak{A}$. The following assertions are equivalent:
  \begin{enumerate}
\item  [$(1)$ \namedlabel{flatpurethe1}{$(1)$}] $\pi_{F}$ is flat;
\item  [$(2)$ \namedlabel{flatpurethe2}{$(2)$}] $F$ is pure.
\end{enumerate}
\end{theorem}
\begin{proof}
  \item [\ref{flatpurethe1}$\Rightarrow$\ref{flatpurethe2}:] Following Proposition \ref{canonflat}, it proves by taking $G=\{1\}$.
  \item [\ref{flatpurethe2}$\Rightarrow$\ref{flatpurethe1}:] Let $G$ be a filter of $\mathfrak{A}$ and $a\in F$. Consider $x\in (G\veebar F:a)$. So there exist $g\in G$ and $f\in F$ such that $g\odot f\leq x\vee a$. By hypothesis there exist $s\in f^{\perp}$ and $t\in F$ such that $s\odot t=0$. By \ref{res2} we have
      \begin{equation}\label{eq1}
        s\vee g\leq s\vee (g\odot f)\leq s\vee (x\vee a).
      \end{equation}
      Also, we
      \begin{equation}\label{eq2}
        x=x\vee (s\odot t)\geq (x\vee s)\odot (x\vee s)\geq (x\vee t)\odot s.
      \end{equation}
      Applying \ref{eq1} and \ref{eq2} implies that $x\in (G:a)\veebar F$. This holds the result.
\end{proof}

  \citet[p. 291]{de1983projectivity} extended the notion of the support of a module to an ideal of a unitary commutative ring.  This generalization allows us to define a similar notion for residuated lattices. Let $F$ be a filter of a residuated lattice $\mathfrak{A}$. The set $\bigcup_{f\in F}h(f^{\perp})$ is called \textit{the support of $F$} and denoted by $Supp(F)$. It is obvious that $d(F)\subseteq Supp(F)$, and $F=\{1\}$ if and only if $Supp(F)=\emptyset$. Also, for any $a\in A$, $Supp(\mathscr{F}(a))=h(a^{\perp})$, and so $Supp(\mathscr{F}(a))$ is a closed set of $Spec_{h}(\mathfrak{A})$. The next theorem characterizes pure filters, inspired by the one obtained for unitary commutative rings by \citet[Proposition 1.5]{de1983projectivity}.
\begin{theorem}\label{pureequalsupport}
  Let $\mathfrak{A}$ be a residuated lattice and $F$ a filter of $\mathfrak{A}$. The following assertions are equivalent:
  \begin{enumerate}
\item  [$(1)$ \namedlabel{pureequalsupport1}{$(1)$}] $d(F)=Supp(F)$;
\item  [$(2)$ \namedlabel{pureequalsupport2}{$(2)$}] $F$ is pure.
\end{enumerate}
\end{theorem}
\begin{proof}
\item [\ref{pureequalsupport1}$\Rightarrow$\ref{pureequalsupport2}:] Let $f\in F$. If $f^{\perp}\veebar F\subsetneq A$, then $f^{\perp}\subseteq \mathfrak{p}$, for some $\mathfrak{p}\in Spec(\mathfrak{A})$. This states that $\mathfrak{p}\in d(F)$; a contradiction
\item [\ref{pureequalsupport2}$\Rightarrow$\ref{pureequalsupport1}:] Let $\mathfrak{p}\in Supp(F)$. So there exists $f\in F$ such that $f^{\perp}\subseteq \mathfrak{p}$. This implies that $F\nsubseteq \mathfrak{p}$ and this shows the equality.
\end{proof}
\begin{theorem}
  Let $F$ be a pure filter of a residuated lattice $\mathfrak{A}$. Then $F=\{a\in A\mid d(a)\subseteq Supp(F)\}$.
\end{theorem}
\begin{proof}
Take $a\in A$ such that $d(a)\subseteq Supp(F)$. Suppose there exists $P\in h(F)$ such that $a\notin F$. By Theorem \ref{pureequalsupport} follows that $P\in d(F)$; a contradiction. So the result holds due to Proposition \ref{intprimfilt}.
\end{proof}

By Theorem \ref{closefalzai} follows that the closed sets of $Spec_{h}(\mathfrak{A})$ are $\mathscr{S}$-stable. It may arise a question, are there open $\mathscr{S}$-stable subsets of $Spec_{h}(\mathfrak{A})$, and, if so, can we characterize them by their determining filters? For commutative rings with unity \citep{de1983projectivity}, bounded distributive lattices \citep{al1990topological}, and MV-algebras \citep{georgescu1997pure,belluce2000stable}, this has given rise to the notion of pure ideals. Also, for residuated lattices \citep[Theorem 61]{bucsneag2012stable}, it is shown that this has created the notion of pure filters.
\begin{remark}
 For MV-algebras \citep{belluce2000stable}, and residuated lattices \citep{bucsneag2012stable}, $\mathscr{S}$-stable subsets of the hull-kernel topology is called \textit{stable under ascent}.
\end{remark}

The following theorem gives a criteria for open $\mathscr{S}$-stable subsets of the spectrum of a residuated lattice by means of pure filters in a different way of \citet[Theorem 61]{bucsneag2012stable}.
\begin{theorem}\label{purestable}
Let $\mathfrak{A}$ be a residuated lattice and $F$ a filter of $\mathfrak{A}$. The following assertions are equivalent:
  \begin{enumerate}
\item  [$(1)$ \namedlabel{purestable1}{$(1)$}] $d(F)$ is $\mathscr{S}$-stable;
\item  [$(2)$ \namedlabel{purestable2}{$(2)$}] $F$ is pure.
\end{enumerate}
\end{theorem}
\begin{proof}
\item [\ref{purestable1}$\Rightarrow$\ref{purestable2}:] Consider $\mathfrak{q}\in \mathscr{G}h(F)$. So $\mathfrak{q}\subseteq \mathfrak{p}$ for some $\mathfrak{p}\notin d(F)$. This implies that $\mathfrak{q}\notin d(F)$ and so the result holds.
\item [\ref{purestable2}$\Rightarrow$\ref{purestable1}:] Consider $\mathfrak{p},\mathfrak{q}\in Spec(\mathfrak{A})$ such that $\mathfrak{p}\in d(F)$ and $\mathfrak{p}\subseteq \mathfrak{q}$. If $F\subseteq \mathfrak{q}$, $\mathfrak{p}\in \mathscr{G}h(F)$; a contradiction.
\end{proof}

\citet[\S 7. Proposition 7]{borceux1983algebra} show that any sum and any finite intersection of pure ideals of a unitary ring (not necessarily commutative) is a pure ideal. The following theorem improves and generalizes the case for residuated lattices.
\begin{theorem}\label{sigmfiltlatt}
   Let $\mathfrak{A}$ be a residuated lattice. The following assertions hold:
    \begin{enumerate}
\item  [$(1)$ \namedlabel{sigmfiltlatt1}{$(1)$}] $(\sigma(\mathfrak{A});\cap,\veebar)$ is a frame;
\item  [$(2)$ \namedlabel{sigmfiltlatt2}{$(2)$}] $(\sigma(\mathfrak{A});\cap,\vee^{\alpha})$ is a frame;
\item  [$(3)$ \namedlabel{sigmfiltlatt3}{$(3)$}] $(\sigma(\mathfrak{A});\cap,\vee^{\omega})$ is a lattice;
\end{enumerate}
\end{theorem}
\begin{proof}
\item [\ref{sigmfiltlatt1}:]
Let $\mathcal{F}$ be a family of pure filters of $\mathfrak{A}$. By Proposition \ref{primesigmad} follows that
\[\veebar \mathcal{F}=\veebar_{F\in \mathcal{F}} \sigma(F)\subseteq \sigma(\veebar \mathcal{F})\subseteq \veebar \mathcal{F}.\]
\item [\ref{sigmfiltlatt2}:] It is evident by \ref{sigmfiltlatt1}.
\item [\ref{sigmfiltlatt3}:] Applying Proposition \ref{sigmafequiv}\ref{sigmafequiv6} shows that $\sigma(\mathfrak{A})\subseteq \Omega(\mathfrak{A})$. Let $F,G\in \sigma(\mathfrak{A})$. We have
    \[F\veebar G\subseteq F\vee^{\omega} G=\omega(I_F) \vee^{\omega} \omega(I_G)\subseteq \omega(I_{F\veebar G})=\sigma(F\veebar G)=F\veebar G.\]
    This holds the result due to \ref{sigmfiltlatt1}.
\end{proof}

The next theorem shows that in a residuated lattice every filter is pure if and only if it is hyperarchimedean.
\begin{theorem}\label{sigmahyper}
  Let $\mathfrak{A}$ be a residuated lattice. The following assertions are equivalent:
  \begin{enumerate}
\item  [$(1)$ \namedlabel{sigmahyper1}{$(1)$}] $\mathscr{PF}(\mathfrak{A})\subseteq \sigma(\mathfrak{A})$;
\item  [$(2)$ \namedlabel{sigmahyper2}{$(2)$}] $\mathfrak{A}$ is hyperarchimedean;
\item  [$(3)$ \namedlabel{sigmahyper3}{$(3)$}] $\mathscr{F}(\mathfrak{A})\subseteq \sigma(\mathfrak{A})$.
\end{enumerate}
\end{theorem}
\begin{proof}
It is an immediate consequence of Proposition \ref{hperarchpri}.
%\item [\ref{sigmahyper1}$\Rightarrow$\ref{sigmahyper2}:] Let $F\in \mathscr{PF}(\mathfrak{A})$. So $F=\mathscr{F}(x)$, for some $x\in A$. By hypothesis we get $x^{\perp}\veebar F=A$ and this holds the result by Propositions \ref{hperarchpri} and \ref{fbetasig}.
%\item [\ref{sigmahyper2}$\Rightarrow$\ref{sigmahyper1}:]
\end{proof}
\begin{proposition}\label{comxpureprime}
  Let $\mathfrak{A}$ be a residuated lattice. Any two distinct elements of the set $Spec(\mathfrak{A})\cap \sigma(\mathfrak{A})$ are comaximal.
\end{proposition}
\begin{proof}
  Let $\mathfrak{p}_1$ and $\mathfrak{p}_2$ be two distinct elements of the set $Spec(\mathfrak{A})\cap \sigma(\mathfrak{A})$. Consider $x\in \mathfrak{p}_1\setminus \mathfrak{p}_2$. So $\mathfrak{p}_1\veebar x^{\perp}=A$ and $x^{\perp}\subseteq \mathfrak{p}_2$. This holds the result.
\end{proof}

Following Theorem \ref{purestable}, the open $\mathscr{S}$-stable subsets of $Spec_{h}(\mathfrak{A})$ are precisely of the form $d(F)$, in which $F$ runs over the pure filters of $\mathfrak{A}$. So, applying Proposition \ref{sigmfiltlatt}, the set $\tau_{\mathscr{D}}=\{d(F)\mid F\in \sigma(\mathfrak{A})\}$ forms a topology on $Spec(\mathfrak{A})$ which is called the $\mathscr{D}$-topology on $\mathfrak{A}$. It is obvious that $\tau_{h}$ is finer than $\tau_{\mathscr{D}}$.
   \begin{remark}
     \cite{lazard1967disconnexites}, for a commutative ring with unity $\mathfrak{A}$, showed that the open $\mathscr{S}$-stable subsets of $Spec(\mathfrak{A})$ forms a topology, coarser than the hull-kernel topology, which is called the $\mathscr{D}$-topology. This line is continued by \cite{de1983projectivity}, for commutative rings with unity, and by \citep{al1990topological}, for bounded distributive lattices. Also, $\mathscr{D}$-topology in MV-algebras \citep{belluce2000stable} and residuated lattices \citep{bucsneag2012stable} is discussed under the name of \textit{stable topology}.
   \end{remark}

The next result gives a characterization for hyperarchimedean residuated lattices using $\mathscr{D}$-topology on $Spec(\mathfrak{A})$.
\begin{theorem}\label{huldtopohyper}
  Let $\mathfrak{A}$ be a residuated lattice.  The hull-kernel topology coincides with the $\mathscr{D}$-topology on $Spec(\mathfrak{A})$ if and only if $\mathfrak{A}$ is hyperarchimedean.
\end{theorem}
\begin{proof}
  It is an immediate consequence of Theorems \ref{purestable} \& \ref{sigmahyper}.
\end{proof}

%%%%%%%%%%%%%%%%%%%%%%%%

Let $\mathfrak{A}$ be a residuated lattice. For any filter $F$ of $\mathfrak{A}$, we set
\[\rho(F)=\underline{\bigvee}\{G\in \sigma(\mathfrak{A})\mid G\subseteq F\},\]
and call it \textit{the pure part of $F$}. The following proposition gives some properties of the pure part of a filter of a residuated lattice.
\begin{proposition}\label{rfilter}
  Let $\mathfrak{A}$ be a residuated lattice. The following assertions hold:
\begin{enumerate}
\item  [$(1)$ \namedlabel{rfilter1}{$(1)$}] $\rho(F)\subseteq \sigma(F)$,  for any $F\in \mathscr{F}(\mathfrak{A})$;
\item  [$(2)$ \namedlabel{rfilter2}{$(2)$}] $\rho(F)$ is the greatest pure filter of $\mathfrak{A}$ contained in $F$;
\item  [$(3)$ \namedlabel{rfilter3}{$(3)$}] the map $\rho$, given by $F\rightsquigarrow \rho(F)$, is a closure operator on $\mathscr{F}(\mathfrak{A})$, with the fixed points $\sigma(\mathfrak{A})$;
\item  [$(4)$ \namedlabel{rfilter4}{$(4)$}] $\rho(F)\cap \rho(G)=\rho(F\cap G)$, for any $F,G\in \mathscr{F}(\mathfrak{A})$;
\item  [$(5)$ \namedlabel{rfilter5}{$(5)$}] $\veebar_{F\in \mathscr{F}}\rho(F)\subseteq \rho(\veebar \mathscr{F})$,for any $\mathscr{F}\subseteq \mathscr{F}(\mathfrak{A})$;
\item  [$(6)$ \namedlabel{rfilter6}{$(6)$}] $F=\bigcap\{\rho(\mathfrak{m})\mid \mathfrak{m}\in h_{M}(F)\}$, for any $F\in \sigma(\mathfrak{A})$;
\item  [$(7)$ \namedlabel{rfilter7}{$(7)$}] $F=\rho(Rad(F))$,  for any $F\in \sigma(\mathfrak{A})$;
\item  [$(8)$ \namedlabel{rfilter8}{$(8)$}] $\rho(\mathfrak{p})=\rho(D(\mathfrak{p}))$,  for any $\mathfrak{p}\in Spec(\mathfrak{A})$.
\end{enumerate}
\end{proposition}
\begin{proof}
\item [\ref{rfilter1}:] It is evident by Proposition \ref{sigmapro}\ref{sigmapro2}.
\item [\ref{rfilter2}:] It is evident by Theorem \ref{sigmfiltlatt}\ref{sigmfiltlatt1}.
\item [\ref{rfilter3}:] It is obvious.
\item [\ref{rfilter4}:] Let $F,G\in \mathscr{F}(\mathfrak{A})$. So $\rho(F\cap G)\subseteq \rho(F)\cap \rho(G)$, since $\rho$ is isotone. Conversely, $\rho(F)\cap \rho(G)\subseteq F\cap G$, and this implies $\rho(F)\cap \rho(G)\subseteq \rho(F\cap G)$.
\item [\ref{rfilter5}:] There is nothing to prove.
\item [\ref{rfilter6}:] It is a direct result of Proposition \ref{sigmafequiv}\ref{sigmafequiv5} and \ref{rfilter1}.
\item [\ref{rfilter7}:] Let $F$ be a pure filter of $\mathfrak{A}$. Since $\rho$ is isotone, so $F\subseteq \rho(Rad(F))$. Conversely, let $\mathfrak{m}$ be an arbitrary maximal filter of $\mathfrak{A}$ containing $F$. We get that $\rho(Rad(F))\subseteq \rho(\mathfrak{m})$, and this establishes the result due to \ref{rfilter4}.
\item [\ref{rfilter8}:] Let $\mathfrak{p}$ be a prime filter of $\mathfrak{A}$. By Proposition \ref{sigmapro}\ref{sigmapro3}, \ref{rfilter1} and \ref{rfilter2} follows that $\rho(\mathfrak{p})\subseteq \rho(\sigma(\mathfrak{p}))\subseteq \rho(D(\mathfrak{p}))\subseteq \rho(\mathfrak{p})$.
\end{proof}
%%%%%%%%%%%%%%%%%%%%%%%%%%%%%%%%%%%%%%%%%%%%%%%%%%%%%%%%%%%%%%%%%%%%%%%%%%%%%%%%%%%%%%%%%%%%%%%%%%%%%%%%%%%%%%%%%%%%%
The next result can be compared with Proposition \ref{filqou}.
\begin{proposition}\label{purefilqou}
Let $f:\mathfrak{A}\longrightarrow \mathfrak{B}$ be an epimorphism. We have
\[\sigma(\mathfrak{B})=\{f(F)\mid coker(f)\subseteq F\in \sigma(\mathfrak{A})\}.\]
\end{proposition}
\begin{proof}
With a little effort, it follows from Proposition \ref{prorh} and Theorem \ref{sigmafequiv}.
\end{proof}

The subsequent upshot characterizes the pure filters of a quotient residuated lattice.
\begin{corollary}\label{corpurefilqout}
Let $\mathfrak{A}$ be a residuated lattice and $F$ a pure filter of $\mathfrak{A}$. We have
\[\sigma(\mathfrak{A}/F)=\{H/F|F\subseteq H\in \sigma(\mathfrak{A})\}.\]
\end{corollary}
\begin{proof}
By considering the canonical projection $\pi_{F}$, it follows from Proposition \ref{purefilqou}.
\end{proof}
%%%%%%%%%%%%%%%%%%%%%%%%%%%%%%%%%%%%%%%%%%%%%%%%%%%%%%%%%%%%%%%%%%%%%%%%%%%%%%%%%%%%%%%%%
\section{The pure spectrum of a residuated lattice}\label{sec4}

The purely-prime notion was introduced and studied in \citet[\S 7,8]{borceux1983algebra}
for general rings (not necessarily commutative).

Let $\mathfrak{A}$ be a residuated lattice. A proper pure filter of $\mathfrak{A}$ is called \textit{purely-maximal} provided that it is a maximal element in the set of proper and pure filters of $\mathfrak{A}$. The set of purely-maximal filters of $\mathfrak{A}$ shall be denoted by $Max(\sigma(\mathfrak{A}))$.  A proper pure filter $\mathfrak{p}$ of $\mathfrak{A}$ is called \textit{purely-prime} provided that $F_1\cap F_2=\mathfrak{p}$ implies $F_1=\mathfrak{p}$ or $F_2=\mathfrak{p}$, for any $F_1,F_2\in \sigma(\mathfrak{A})$. The set of all purely-prime filters of $\mathfrak{A}$ shall be denoted by $Spp(\mathfrak{A})$.  It is obvious that $Max(\sigma(\mathfrak{A}))\subseteq Spp(\mathfrak{A})$. Zorn's lemma ensures that any proper pure filter is contained in a purely-maximal filter, and so in a purely-prime filter.
\begin{proposition}\label{preqpropu}
Let $\mathfrak{A}$ be a residuated lattice and $\mathfrak{p}$ a proper pure filter of $\mathfrak{A}$. The following assertions are equivalent:
\begin{enumerate}
\item  [(1) \namedlabel{preqpropu1}{(1)}] $\mathfrak{p}$ is purely-prime;
\item  [(2) \namedlabel{preqpropu2}{(2)}] $F_1\cap F_2\subseteq \mathfrak{p}$ implies $F_1\subseteq \mathfrak{p}$ or $F_2\subseteq \mathfrak{p}$, for any $F_1,F_2\in \sigma(\mathfrak{A})$.
\end{enumerate}
\end{proposition}
\begin{proof}
  By distributivity of the lattice $\sigma(\mathfrak{A})$, it is evident.
\end{proof}

\begin{proposition}\label{spurefilqou}
Let $f:\mathfrak{A}\longrightarrow \mathfrak{B}$ be an epimorphism. We have
\[Spp(\mathfrak{B})=\{f(\mathfrak{p})\mid coker(f)\subseteq \mathfrak{p}\in Spp(\mathfrak{A})\}.\]
\end{proposition}
\begin{proof}
With a little effort, it follows from Propositions \ref{purefilqou} \& \ref{preqpropu}.
\end{proof}

The next corollary characterizes the purely-prime filters of a quotient residuated lattice.
\begin{corollary}\label{corpurefilqout}
Let $\mathfrak{A}$ be a residuated lattice and $F$ a pure filter of $\mathfrak{A}$. We have
\[Spp(\mathfrak{A}/F)=\{\mathfrak{p}/F|F\subseteq \mathfrak{p}\in Spp(\mathfrak{A})\}.\]
\end{corollary}
\begin{proof}
By considering the canonical projection $\pi_{F}$, it follows from Proposition \ref{spurefilqou}.
\end{proof}

\begin{proposition}\label{comppurpri}
   Let $\mathfrak{p}$ be a purely-prime filter of a residuated lattice $\mathfrak{A}$ and $e\in \beta(\mathfrak{A})$. $\mathfrak{p}$ contains precisely one of $e$ or $\neg e$.
\end{proposition}
\begin{proof}
  It is an immediate consequence of Propositions \ref{genfilprop}\ref{genfilprop3} \& \ref{boleleprop}\ref{boleleprop2}, Theorem \ref{fbetasig}, and Proposition \ref{preqpropu}.
\end{proof}

\begin{proposition}\label{r1filter}
  Let $\mathfrak{A}$ be a residuated lattice. The following assertions hold:
\begin{enumerate}
\item  [$(1)$ \namedlabel{r1filter1}{$(1)$}] $Max(\sigma(\mathfrak{A}))\subseteq \rho(Max(\mathfrak{A}))$;
\item  [$(2)$ \namedlabel{r1filter2}{$(2)$}] $\rho(Spec(\mathfrak{A}))\subseteq Spp(\mathfrak{A})$;
\item  [$(3)$ \namedlabel{r1filter3}{$(3)$}] $F=\bigcap \{\mathfrak{p}\mid F\subseteq \mathfrak{p}\in Spp(\mathfrak{A})\}$, for any pure filter $F$ of $\mathfrak{A}$.
\end{enumerate}
\end{proposition}
\begin{proof}
\item [\ref{r1filter1}:] There is nothing to prove.
\item [\ref{r1filter2}:] It is straightforward by Proposition \ref{rfilter}\ref{rfilter2}.
\item [\ref{r1filter3}:] Let $\Phi=\{\mathfrak{p}\mid F\subseteq \mathfrak{p}\in Spp(\mathfrak{A})\}$ and $\Psi=\{\rho(\mathfrak{m})\mid \mathfrak{m}\in h_{M}(F)\}$. Using Proposition \ref{rfilter}\ref{rfilter6} and \ref{r1filter2}, it follows that $F\subseteq \bigcap\Phi\subseteq \bigcap\Psi=F$.
\end{proof}
%\begin{lemma}\label{purprielepri}
%Let $\mathfrak{A}$ be a residuated lattice, $F$ a pure filter of $\mathfrak{A}$, and $a\in F$. If for any purely-prime filter $P$ of $\mathfrak{A}$ which contains $a$ we have $F\subseteq P$, then $F=\mathscr{F}(a)$.
%\end{lemma}
%\begin{proof}
%It is a direct consequence of Proposition \ref{r1filter}\ref{r1filter3}.
%\end{proof}
Let $\mathfrak{A}$ be a residuated lattice. A purely-prime filter $\mathfrak{p}$ of $\mathfrak{A}$ is called  \textit{purely-minimal prime} provided that $\mathfrak{p}$ is a minimal element in the set of all purely-prime filters of $\mathfrak{A}$. The following theorem can be compared with Proposition \ref{mp}.
\begin{corollary}\label{minpurfil}
  Any purely-prime filter of a residuated lattice contains a purely-minimal prime filter.
\end{corollary}
\begin{proof}
  Let $\mathfrak{A}$ be a residuated lattice and $\mathfrak{p}\in Spp(\mathfrak{A})$. Let $\Sigma$ be the set of all purely-prime filters of $\mathfrak{A}$ which are contained in $\mathfrak{p}$. If $\mathfrak{C}$ is a non-empty chain in $\Sigma$, then it is easy to see that $\rho(\bigcap\mathfrak{C})\in \Sigma$. Therefore, by the Zorn’s lemma, $\Sigma$ has at least a minimal element.
\end{proof}

In the theory of commutative rings, a basic interesting and useful result of I. S. Cohen is that for a ring to be Noetherian i.e. for every ideal to be finitely generated, it is enough to have
that every prime ideal be finitely generated \cite[p. 388]{hungerford1974algebra}. \citet[Theorem 3.23]{rasouli2019going} show that in a residuated lattice every prime filter is principal if and only if every filter is principal. \citet[Theorem 6.2]{tarizadeh2021purely} called a commutative ring $\mathfrak{A}$ semi-Noetherian provided that every pure ideal of $\mathfrak{A}$ is finitely generated, and established that if every purely-maximal ideal of a commutative ring is finitely generated, then it is semi-Noetherian. The following theorem improves and generalizes the case for residuated lattices.
\begin{lemma}\label{prinpuregen}
  Every principal pure filter of a residuated lattice is generated by a complemented element.
\end{lemma}
\begin{proof}
  It is an immediate consequence of Proposition \ref{b9fxpro}.
\end{proof}
\begin{theorem}(Cohen type theorem)\label{purpriprithe}
Let $\mathfrak{A}$ be a residuated lattice. The following assertions are equivalent:
\begin{enumerate}
\item  [(1) \namedlabel{purpriprithe1}{(1)}] each purely-maximal filter of $\mathfrak{A}$ is principal;
\item  [(2) \namedlabel{purpriprithe2}{(2)}] each purely-prime filter of $\mathfrak{A}$ is principal;
\item  [(3) \namedlabel{purpriprithe3}{(3)}] each pure filter of $\mathfrak{A}$ is principal.
\end{enumerate}
\end{theorem}
\begin{proof}
  \item [\ref{purpriprithe1}$\Rightarrow$\ref{purpriprithe2}:] It follows by Lemma \ref{prinpuregen}.
  \item [\ref{purpriprithe2}$\Rightarrow$\ref{purpriprithe3}:] Set $\Psi$ be the set of all non-principal pure filters of $\mathfrak{A}$. Let $\Psi$ be a non-empty set. By Zorn's lemma $\Psi$ has a maximal element $\mathfrak{m}$. Obviously, $\mathfrak{m}$ is proper. Assume that $F_1$ and $F_2$ are two pure filters of $\mathfrak{A}$ such that $F_1\cap F_2\subseteq \mathfrak{m}$. Suppose that $F_1\nsubseteq \mathfrak{m}$ and $F_2\nsubseteq \mathfrak{m}$. By Theorem \ref{sigmfiltlatt}, $\mathfrak{m}\veebar F_{1}$ and $\mathfrak{m}\veebar F_{2}$ are pure filters. So $\mathfrak{m}\veebar F_{1}=\mathscr{F}(a)$ and $\mathfrak{m}\veebar F_{2}=\mathscr{F}(b)$, for some $a,b\in A$. Applying Proposition \ref{genfilprop}\ref{genfilprop3} and distributivity of $\mathscr{F}(\mathfrak{A})$, it follows that $\mathfrak{m}=\mathscr{F}(a\vee b)$; a contradiction. So $\mathfrak{m}$ is purely-prime; a contradiction, too.
  \item [\ref{purpriprithe3}$\Rightarrow$\ref{purpriprithe1}:] It is evident.
\end{proof}

Let $\mathfrak{A}$ be a residuated lattice. For each pure filter $F$ of $\mathfrak{A}$ we set $d_{\kappa}(F)=\{\mathfrak{p}\in Spp(\mathfrak{A})\mid F\nsubseteq \mathfrak{p}\}$. We can topologize $Spp(\mathfrak{A})$  by taking the set $\{d_{\kappa}(F)\mid F\in \sigma(\mathfrak{A})\}$ as the open sets. One can easily verify the following properties:
\begin{itemize}
  \item $d_{\kappa}(\{1\})=\emptyset$ and $d_{\kappa}(A)=Spp(\mathfrak{A})$;
  \item $d_{\kappa}(F_1)\cap d_{\kappa}(F_2)=d_{\kappa}(F_1\cap F_2)$, for any $F_1,F_2\in \sigma(\mathfrak{A})$;
  \item $\bigcup_{F\in \mathcal{F}} d_{\kappa}(F)=d_{\kappa}(\veebar \mathcal{F})$, for any $\mathcal{F}\subseteq \sigma(\mathfrak{A})$.
\end{itemize}

The set $Spp(\mathfrak{A})$ endowed with this topology shall be called the \textit{pure spectrum} of $\mathfrak{A}$. It is obvious that the closed subsets of the pure spectrum $Spp(\mathfrak{A})$ are precisely of the form $h_{\kappa}(F) =\{\mathfrak{p}\in Spp(\mathfrak{A})\mid F\subseteq \mathfrak{p}\}$, where $F$ runs over pure filters of $\mathfrak{A}$. In the sequel, if $e$ is a complemented element of $\mathfrak{A}$, then we shall denote $d_{\kappa}(\mathscr{F}(e))$ simply by $d_{\kappa}(e)$.

The following theorem shows that the pure spectrum of a residuated lattice is a $T_0$ space.
\begin{theorem}\label{t0spppacecon}
Let $\mathfrak{A}$ be a residuated lattice. $Spp(\mathfrak{A})$ is a $T_0$ spaces.
\end{theorem}
\begin{proof}
Let $\mathfrak{p}$ and $\mathfrak{q}$ be two disjoint purely-prime filter of $\mathfrak{A}$. Suppose that $\mathfrak{p}\nsubseteq \mathfrak{q}$. So $d_{\kappa}(\mathfrak{p})$ is a neighborhood of $\mathfrak{q}$ not containing $\mathfrak{p}$.
\end{proof}

Let $\mathfrak{A}$ be a residuated lattice. In the following, for a given subset $\pi$ of $Spp(\mathfrak{A})$, $cl_{\kappa}(\pi)$ stands for the closure of $\pi$ in the topological space $Spp(\mathfrak{A})$. If $\pi=\{\mathfrak{p}\}$, for some purely-prime filter $\mathfrak{p}$ of $\mathfrak{A}$, then $cl_{\kappa}(\{\mathfrak{p}\})$ is simply denoted by $cl_{\kappa}(\mathfrak{p})$.
\begin{lemma}\label{closurofp}
Let $\mathfrak{A}$ be a residuated lattice and $\mathfrak{p}$ a purely-prime filter of $\mathfrak{A}$. Then $cl_{\kappa}(\mathfrak{p})=h_{\kappa}(\mathfrak{p})$.
\end{lemma}
\begin{proof}
Let $cl_{\kappa}(\mathfrak{p})=h_{\kappa}(F)$, for some pure filter $F$ of $\mathfrak{A}$. Assume that $\mathfrak{q}\in h_{\kappa}(\mathfrak{p})$. Thus $\mathfrak{p}\subseteq \mathfrak{q}$, and this implies that $\mathfrak{q}\in h_{\kappa}(F)$. The converse is evident.
\end{proof}

The next theorem gives a necessary and sufficient condition for the pure spectrum of a residuated lattice to be a $T_1$ space. Recall that a topological space is $T_1$ if and only if points are closed.
\begin{theorem}\label{t1spaspp}
Let $\mathfrak{A}$ be a residuated lattice. $Spp(\mathfrak{A})$ is $T_{1}$ if and only if it is an antichain.
\end{theorem}
\begin{proof}
It is an immediate consequence of Lemma \ref{closurofp}.
\end{proof}

\begin{proposition}\label{sigmad}
  Let $\mathfrak{A}$ be a residuated lattice. The mapping $\delta:\sigma(\mathfrak{A})\longrightarrow \{d_{\kappa}(F)\mid F\in \sigma(\mathfrak{A})\}$, defined by $\delta(F)=d_{\kappa}(F)$, for any $F\in \sigma(\mathfrak{A})$, is a lattice isomorphism.
\end{proposition}
\begin{proof}
It is evident that $\delta$ is a surjection. The injectivity follows by Proposition \ref{r1filter}\ref{r1filter3}.
\end{proof}
\begin{theorem}\label{puresppcomp}
  Let $\mathfrak{A}$ be a residuated lattice. $Spp(\mathfrak{A})$ is a compact space.
\end{theorem}
\begin{proof}
  It is an immediate consequence of Proposition \ref{sigmad}.
\end{proof}

Following by \cite[p. 94]{bourbaki1972commutative}, a non-empty subset $X$
of a topological space $A$ is called \textit{irreducible} if for every pair
of closed subsets $C_1$ and $C_2$ in $A$, $X\subseteq C_1\cup C_2$ implies $X\subseteq C_1$ or $X\subseteq C_2$. Equivalently, all non empty open subsets of X are dense or any two nonempty open sets have nonempty intersection.
\begin{proposition}\label{irrsppdclosub}
  Let $\mathfrak{A}$ be a residuated lattice. The following assertions are equivalent:
\begin{enumerate}
   \item [(1) \namedlabel{irrsppdclosub1}{(1)}] $C$ is an irreducible closed subset of $Spp(\mathfrak{A})$;
   \item [(2) \namedlabel{irrsppdclosub2}{(2)}] $C=h_{\kappa}(\mathfrak{p})$, for some purely-prime filter $\mathfrak{p}$ of $\mathfrak{A}$.
 \end{enumerate}
\end{proposition}
\begin{proof}
\item [\ref{irrsppdclosub1}$\Rightarrow$\ref{irrsppdclosub2}:] Let $C$ be an irreducible closed subset of $Spp(\mathfrak{A})$. So $C=h(F)$, for some pure filter $F$ of $\mathfrak{A}$. Let $F_1\cap F_2\subseteq F$, for some pure filters $F_1$ and $F_2$ of $\mathfrak{A}$. Using Proposition \ref{sigmad}, it follows that $C=(h_{\kappa}(F)\cap h_{\kappa}(F_1))\cup (h_{\kappa}(F)\cap h_{\kappa}(F_2))$. Without loss of generality, assume that $C\subseteq h_{\kappa}(F)\cap h_{\kappa}(F_1)$. So $F_1\subseteq F$.
\item [\ref{irrsppdclosub2}$\Rightarrow$\ref{irrsppdclosub1}:] Let $\mathfrak{p}$ be a purely-prime filter of $\mathfrak{A}$. Let $h_{\kappa}(\mathfrak{p})\subseteq h_{\kappa}(F_1)\cup h_{\kappa}(F_2)=h_{\kappa}(F_1\cap F_2)$. Therefore, $F_1\cap F_2\subseteq \mathfrak{p}$, and so $F_1\subseteq \mathfrak{p}$ or $F_2\subseteq \mathfrak{p}$.
\end{proof}

Following \cite[p. 150]{wraith1975lectures}, a topological space $A$ is called \textit{Sober} provided that every irreducible closed subset $\mathcal{C}$ of $A$ has a unique \textit{generic point}, i.e. $\mathcal{C}=cl(c)$, for some $c\in \mathcal{C}$. Any Hausdorff space is sober, and all sober spaces are $T_{0}$, and both implications are strict. Sobriety is not comparable to the $T_{1}$ condition, and moreover $T_{2}$ is stronger than $T_{1}$ and sober.
\begin{theorem}\label{Soberspec}
  Let $\mathfrak{A}$ be a residuated lattice. Then $Spp(\mathfrak{A})$ is a Sober space.
 \end{theorem}
\begin{proof}
 It is an immediate consequence of Theorem \ref{t0spppacecon} and Proposition \ref{irrsppdclosub}.
\end{proof}

Proposition \ref{r1filter}\ref{r1filter2} leads us to a map $\rho:Spec(\mathfrak{A})\longrightarrow Spp(\mathfrak{A})$, defined by $\mathfrak{p}\rightsquigarrow \rho(\mathfrak{p})$, which is called \textit{the pure part map of $\mathfrak{A}$}. In general, this map is not injective or even surjective. However, by Proposition \ref{r1filter}\ref{r1filter1} follows that if $\mathfrak{n}$ is a purely-maximal filter of $\mathfrak{A}$, then $\rho(\mathfrak{m})=\mathfrak{n}$, for some $\mathfrak{m}\in Max(\mathfrak{A})$.
\begin{theorem}\label{spsppconti}
   Let $\mathfrak{A}$ be a residuated lattice. The pure part map of $\mathfrak{A}$ is continuous.
\end{theorem}
\begin{proof}
Let $F$ be a pure filter of $\mathfrak{A}$. Routinely, one can show that $\rho^{\leftarrow}(d_{\kappa}(F))=d(F)$.
\end{proof}

A fundamental result due to Grothendieck states that there is a canonical bijection between the idempotents of a ring and the clopen of its prime spectrum. The subsequent outcome generalizes the Grothendieck correspondence to the complemented elements and the clopen subsets of the pure spectrum of a residuated lattice.
\begin{theorem}(Grothendieck type theorem)\label{grothfundres}\label{clopsppbij}
Let $\mathfrak{A}$ be a residuated lattice. The map $\Gamma: e\rightsquigarrow d_{\kappa}(e)$ is a one-to-one correspondence from the set $\beta(\mathfrak{A})$ onto the set $Clop(Spp(\mathfrak{A}))$.
\end{theorem}
\begin{proof}
  Applying Propositions \ref{comppurpri} \& \ref{r1filter}\ref{r1filter3}, it shows that $\Gamma$ is a well-defined injection. Let $\mathscr{C}$ is a clopen subset of $Spp(\mathfrak{A})$. So there exist two pure filters $F,G$ of $\mathfrak{A}$ such that $\mathscr{C}=h_{\kappa}(F)$ and $Spp(\mathfrak{A})\setminus \mathscr{C}=h_{\kappa}(G)$. This implies that $h_{\kappa}(F\cap G)=h_{\kappa}(F)\cup h_{\kappa}(G)=Spp(\mathfrak{A})$ and $h_{\kappa}(F\veebar G)=h_{\kappa}(F)\cap h_{\kappa}(G)=\emptyset$. By Proposition \ref{r1filter}\ref{r1filter3}, we have $F\cap G=\{1\}$ and $F\veebar G=A$. By Proposition \ref{b9fxpro}, there exists $e\in \beta(\mathfrak{A})$ such that $F=\mathscr{F}(e)$ and $G=\mathscr{F}(\neg e)$. So $h_{\kappa}(F)\subseteq h_{\kappa}(e)$ and $h_{\kappa}(G)\subseteq h_{\kappa}(\neg e)$. This implies that $\Gamma(\neg e)=\mathscr{C}$.
\end{proof}

Recall that a topological space $(A;\tau)$ is said to be \textit{disconnected} provided that it is the union of two disjoint non-empty open sets. Otherwise, $(A;\tau)$ is said to be \textit{connected}. A subset of a topological space is said to be \textit{connected} if it is connected under its subspace topology. It is well-known that $(A;\tau)$ is connected if and only if the empty set and the whole space are the only clopen subsets of the space $A$, see \citep[Theorem 6.1.1]{engelking1989general}.
\begin{corollary}\label{sppconn}
  Let $\mathfrak{A}$ be a residuated lattice. $\mathfrak{A}$ is directly indecomposable if and only if $Spp(\mathfrak{A})$ is connected.
\end{corollary}
\begin{proof}
  It is evident by  Proposition \ref{direcindbeta}, and Theorem \ref{clopsppbij}.
\end{proof}

\begin{proposition}\label{contsppfunh}
  The map which sends every residuated lattice $\mathfrak{A}$ to $Spp(\mathfrak{A})$, and every residuated morphism $f:\mathfrak{A}\longrightarrow \mathfrak{B}$ to $Spp(f):Spp(\mathfrak{B})\longrightarrow Spp(\mathfrak{A})$, given by $\mathfrak{p}\rightsquigarrow \rho(f^{\leftarrow}(\mathfrak{p}))$, is a contravariant functor from the category of residuated lattices to the category of topological spaces.
\end{proposition}
\begin{proof}
The contravariant functoriality has a routine verification. Let $f:\mathfrak{A}\longrightarrow \mathfrak{B}$ be a residuated morphism. Consider $F\in \sigma(\mathfrak{A})$. One can see that $Spp(f)^{\leftarrow}(h_{\kappa}(F))=h_{\kappa}(\rho(\mathscr{F}(f(F))))$. This holds the result.
\end{proof}

Notice that if $\mathfrak{A}$ is a residuated lattice and $F$ a pure filter of $\mathfrak{A}$, by Corollary \ref{corpurefilqout}, $Spp(\mathfrak{A}/F)=\{\mathfrak{p}/F|\mathfrak{p}\in h_{\kappa}(F)\}$. The subsequent corollary shows that $Spp(\mathfrak{A}/F)$ and $h_{\kappa}(F)$ are homeomorphic as topological spaces.
\begin{corollary}\label{qoepuruspec}
  Let $\mathfrak{A}$ be a residuated lattice and $F$ a pure filter of $\mathfrak{A}$. The canonical projection $\pi_{F}$ induces a homeomorphism from the pure prime spectrum $Spp(\mathfrak{A}/F)$ onto $h_{\kappa}(F)$.
\end{corollary}
\begin{proof}
Using Proposition \ref{contsppfunh}, it follows that $\pi^{\kappa}$ is a continuous injection. We can show, routinely, that $Im(\pi^{\kappa})=h_{\kappa}(F)$, and the map $Spp(\mathfrak{A}/F)\longrightarrow h_{\kappa}(F)$ is closed.
\end{proof}
%%%%%%%%%%%%%%%%%%%%%%%%%%%%%%%%%%%%%%%%%%%%%%%%%%%%%%%%%%%%%%%%%%%%%%%%%%%%%%%%%%%%%%%%%
\section{The pure spectrum of a Gelfand residuated lattice}\label{sec5}

This section deals with the pure spectrum of a Gelfand residuated lattice. For the basic facts concerning Gelfand residuated lattices, the reader is referred to \cite{rasouli2022gelfand}.
\begin{definition}
A residuated lattice $\mathfrak{A}$ is called \textit{Gelfand} provided that any prime filter of $\mathfrak{A}$ is contained in a unique maximal filter of $\mathfrak{A}$.
\end{definition}
\begin{example}\label{quanorexas}
One can see that the residuated lattice $\mathfrak{A}_6$ from Example \ref{exa6} is not Gelfand, and the residuated lattice $\mathfrak{B}_6$ from Example \ref{exb6}, the residuated lattice $\mathfrak{C}_6$ from Example \ref{exc6}, and the residuated lattice $\mathfrak{A}_8$ from Example \ref{exa8} are Gelfand.
\end{example}
\begin{example}
   The class of MTL-algebras, and so,  MV-algebras, BL-algebras, and Boolean algebras are some subclasses of Gelfand residuated lattices.
\end{example}

The following proposition gives some algebraic characterizations for Gelfand residuated lattices.
\begin{proposition}\cite[Theorem 3.7]{rasouli2022gelfand}\label{pmprop}
Let $\mathfrak{A}$ be a residuated lattice.  The following assertions are equivalent:
 \begin{enumerate}
   \item [(1) \namedlabel{pmprop1}{(1)}] $\mathfrak{A}$ is Gelfand;
   %\item [(2) \namedlabel{pmprop2}{(2)}] for any distinct maximal filters $\mathfrak{m}$ and $\mathfrak{n}$ of $\mathfrak{A}$, there exist $a\notin \mathfrak{m}$ and $b\notin \mathfrak{n}$ such that $a\vee b=1$;
   \item [(3) \namedlabel{pmprop3}{(3)}] for any distinct maximal filters $\mathfrak{m}$ and $\mathfrak{n}$ of $\mathfrak{A}$, $D(\mathfrak{m})$ and $D(\mathfrak{n})$ are comaximal;
%   \item [(4) \namedlabel{pmprop4}{(4)}] for any distinct maximal filters $\mathfrak{m}$ and $\mathfrak{n}$ of $\mathfrak{A}$, there exists $a\in A$ such that $a\in D(\mathfrak{m})$ and $\neg a\in D(\mathfrak{n})$;
%   \item [(5) \namedlabel{pmprop5}{(5)}] for any maximal filter $\mathfrak{m}$ and any prime filter $\mathfrak{p}$ of $\mathfrak{A}$, $D(\mathfrak{p})\subseteq \mathfrak{m}$ implies $\mathfrak{p}\subseteq \mathfrak{m}$;
%   \item [(6) \namedlabel{pmprop6}{(6)}] for any maximal filter $\mathfrak{m}$ and any proper filter $F$ of $\mathfrak{A}$, $D(\mathfrak{m})\subseteq F$ implies $F\subseteq \mathfrak{m}$;
   \item [(7) \namedlabel{pmprop7}{(7)}] for any maximal filter $\mathfrak{m}$ and any proper filter $F$ of $\mathfrak{A}$, $F\veebar \mathfrak{m}=A$ implies $F\veebar D(\mathfrak{m})=A$;
%   \item [(8) \namedlabel{pmprop8}{(8)}] for any maximal filter $\mathfrak{m}$ of $\mathfrak{A}$, $\mathfrak{m}$ is the unique maximal filter containing $D(\mathfrak{m})$;
%   \item [(9) \namedlabel{pmprop9}{(9)}] for any maximal filter $\mathfrak{m}$, $\mathscr{G}(\mathfrak{m})=h(D(\mathfrak{m}))$;
%   \item [(10) \namedlabel{pmprop10}{(10)}] for any maximal filter $\mathfrak{m}$ of $\mathfrak{A}$, $\mathfrak{A}/D(\mathfrak{m})$ is local;
%   \item [(11) \namedlabel{pmprop11}{(11)}] for any maximal filter $\mathfrak{m}$ of $\mathfrak{A}$ and any $x\notin \mathfrak{m}$, there exist an integer $n$ and $a\notin \mathfrak{m}$ such that $a\vee \neg x^{n}=1$.
 \end{enumerate}
\end{proposition}
Recall that a \textit{retraction} is a continuous mapping from a topological space into a subspace which preserves the position of all points in that subspace.
%\citet[Proposition 6.4]{georgescu2015algebraic}
\begin{proposition}\cite[Theorem 3.9]{rasouli2022gelfand}\label{gelnor}
Let $\mathfrak{A}$ be a residuated lattice. The following assertions are equivalent:
\begin{enumerate}
   \item [(1) \namedlabel{gelnor1}{(1)}] $\mathfrak{A}$ is Gelfand;
   \item [(2) \namedlabel{gelnor2}{(2)}] $Max_{h}(\mathfrak{A})$ is a retract of $Spec_{h}(\mathfrak{A})$.
 \end{enumerate}
\end{proposition}

Let $F$ be a filter of a residuated lattice $\mathfrak{A}$. The intersection of all maximal filters of $\mathfrak{A}$ containing $F$ is denoted by $Rad(F)$. The subsequent theorem characterizes the Gelfand residuated lattices using the sink of a filter.
\begin{theorem}\label{equgelchaunit}
  Let $\mathfrak{A}$ be a residuated lattice. The following assertions are equivalent:
\begin{enumerate}
\item  [(1) \namedlabel{equgelchaunit1}{(1)}] $\mathfrak{A}$ is Gelfand;
\item  [(2) \namedlabel{equgelchaunit2}{(2)}] $\sigma(F)\subseteq \mathfrak{m}$ implies $F\subseteq \mathfrak{m}$, for any $F\in \mathscr{F}(\mathfrak{A})$ and $\mathfrak{m}\in Max(\mathfrak{A})$;
\item  [(3) \namedlabel{equgelchaunit3}{(3)}] $h_{M}(F)=h_{M}(\sigma(F))$, for any filter $F$ of $\mathfrak{A}$;
\item  [(4) \namedlabel{equgelchaunit4}{(4)}] $Rad(F)=Rad(\sigma(F))$, for any filter $F$ of $\mathfrak{A}$;
\item  [(5) \namedlabel{equgelchaunit5}{(5)}] if $F$ and $G$ are comaximal filters of $\mathfrak{A}$, then so are $\sigma(F)$ and $\sigma(G)$;
\item  [(6) \namedlabel{equgelchaunit6}{(6)}]  $\veebar_{F\in \mathcal{F}}\sigma(F)=\sigma(\veebar \mathcal{F})$, for any $\mathcal{F}\subseteq \mathscr{F}(\mathfrak{A})$.
\end{enumerate}
\end{theorem}
\begin{proof}
\item [\ref{equgelchaunit1}$\Rightarrow$\ref{equgelchaunit2}:] Let $F$ be a filter of $\mathfrak{A}$ such that $\sigma(F)$ contained in a maximal filter $\mathfrak{m}$ of $\mathfrak{A}$. By absurdum, let $F$ not contained in $\mathfrak{m}$. So $F\veebar \mathfrak{m}=A$, and so by Propositions \ref{sigmapro}\ref{sigmapro4} \& \ref{pmprop}\ref{pmprop7} follows that $F\veebar \sigma(\mathfrak{m})=A$. Hence, there exist $f\in F$ and $a\in \sigma(\mathfrak{m})$ such that $f\odot a=0$. This concludes that $\neg\neg a\in (\neg x)^{\perp}$, for some $x\in \mathfrak{m}$. So $\neg x\in \sigma(F)$; a contradiction.
\item [\ref{equgelchaunit2}$\Rightarrow$\ref{equgelchaunit3}:] It is obvious.
\item [\ref{equgelchaunit3}$\Rightarrow$\ref{equgelchaunit4}:] It is obvious.
\item [\ref{equgelchaunit4}$\Rightarrow$\ref{equgelchaunit5}:] Let $F$ and $G$ be two comaximal filter of $\mathfrak{A}$. Assume by absurdum that $\sigma(F)$ and $\sigma(G)$ are not comaximal. Thus $\sigma(F),\sigma(G)\subseteq \mathfrak{m}$, for some maximal filter $\mathfrak{m}$ of $\mathfrak{A}$. Thus $Rad(\sigma(F)),Rad(\sigma(G))\subseteq \mathfrak{m}$, and this yields that $Rad(F),Rad(G)\subseteq \mathfrak{m}$. This establishes that $F,G\subseteq \mathfrak{m}$; a contradiction.
\item [\ref{equgelchaunit5}$\Rightarrow$\ref{equgelchaunit6}:] Let $\mathcal{F}$ be a family of filters of $\mathfrak{A}$. Consider $a\in \sigma(\veebar \mathcal{F})$. So $(\veebar \mathcal{F})\veebar a^{\perp}=A$. Routinely, one can show that $(\veebar \mathcal{G})\veebar a^{\perp}=A$, for a finite subset $\mathcal{G}$ of $\mathcal{F}$. This results that $(\veebar_{G\in \mathcal{G}}\sigma(G))\veebar a^{\perp}=A$. Therefore, $a\in \sigma(\veebar_{G\in \mathcal{G}}\sigma(G))\subseteq \veebar_{F\in \mathcal{F}}\sigma(F)$.
\item [\ref{equgelchaunit6}$\Rightarrow$\ref{equgelchaunit1}:] It is an immediate consequence of Propositions \ref{sigmapro}\ref{sigmapro4} \& \ref{pmprop}\ref{pmprop3}.
\end{proof}

The following outcome shows that for a given filter of a Gelfand residuated lattice its pure part and its sink are equal.
\begin{corollary}\label{rhosigmanorg}
  Let $\mathfrak{A}$ be a Gelfand residuated lattice and $F$ a filter of $\mathfrak{A}$. Then $\rho(F)=\sigma(F)$.
\end{corollary}
\begin{proof}
  Let $a\in \sigma(F)\setminus \sigma(\sigma(F))$. So $\sigma(F)\veebar a^{\perp}\neq A$. Assume that $\sigma(F)\veebar a^{\perp}\subseteq \mathfrak{m}$, for a maximal filter $\mathfrak{m}$ of $\mathfrak{A}$. By Theorem \ref{equgelchaunit} follows that $F\subseteq \mathfrak{m}$; a contradiction.
\end{proof}

Gelfand rings are characterized in terms of pure ideals by \citet[\S 7, Theorems 31]{borceux1983algebra}. In the following theorem, these results have been improved and generalized to residuated lattices.
\begin{theorem}\label{equgelchapure}
  Let $\mathfrak{A}$ be a residuated lattice. The following assertions are equivalent:
\begin{enumerate}
\item  [(1) \namedlabel{equgelchapure1}{(1)}] $\mathfrak{A}$ is Gelfand;
\item  [(2) \namedlabel{equgelchapure2}{(2)}] $\rho(F)\subseteq \mathfrak{m}$ implies $F\subseteq \mathfrak{m}$, for any $F\in \mathscr{F}(\mathfrak{A})$ and $\mathfrak{m}\in Max(\mathfrak{A})$;
\item  [(3) \namedlabel{equgelchapure3}{(3)}] $h_{M}(F)=h_{M}(\rho(F))$, for any filter $F$ of $\mathfrak{A}$;
\item  [(4) \namedlabel{equgelchapure4}{(4)}] $Rad(F)=Rad(\rho(F))$, for any filter $F$ of $\mathfrak{A}$;
\item  [(5) \namedlabel{equgelchapure5}{(5)}] if $F$ and $G$ are comaximal filters of $\mathfrak{A}$, then so are $\rho(F)$ and $\rho(G)$;
\item  [(6) \namedlabel{equgelchapure6}{(6)}] if $\mathcal{F}$ is a family of filters of $\mathfrak{A}$, then $\rho(\veebar\mathcal{F})=\veebar_{F\in \mathcal{F}}\rho(F)$;
\item  [(7) \namedlabel{equgelchapure7}{(7)}] if $\mathfrak{m}$ and $\mathfrak{n}$ are maximal filters of $\mathfrak{A}$, then $\rho(\mathfrak{m})$ and $\rho(\mathfrak{n})$ are comaximal.

\end{enumerate}
\end{theorem}
\begin{proof}
\item [\ref{equgelchapure1}$\Rightarrow$\ref{equgelchapure2}:] It follows by Theorem \ref{equgelchaunit} and Proposition \ref{rhosigmanorg}.
\item [\ref{equgelchapure2}$\Rightarrow$\ref{equgelchapure3}:] It is evident.
\item [\ref{equgelchapure3}$\Rightarrow$\ref{equgelchapure4}:] It is evident.
\item [\ref{equgelchapure4}$\Rightarrow$\ref{equgelchapure5}:] Let $F$ and $G$ be two comaximal filter of $\mathfrak{A}$. Assume by absurdum that $\rho(F)$ and $\rho(G)$ are not comaximal. Thus $\rho(F),\rho(G)\subseteq \mathfrak{m}$, for some maximal filter $\mathfrak{m}$ of $\mathfrak{A}$. Thus $Rad(\rho(F)),Rad(\rho(G))\subseteq \mathfrak{m}$, and this yields that $Rad(F),Rad(G)\subseteq \mathfrak{m}$. This establishes that $F,G\subseteq \mathfrak{m}$; a contradiction.
\item [\ref{equgelchapure5}$\Rightarrow$\ref{equgelchapure6}:] Let $\mathcal{F}$ be a family of filters of $\mathfrak{A}$. Consider $a\in \rho(\veebar \mathcal{F})$. So $(\veebar \mathcal{F})\veebar a^{\perp}=A$. Routinely, one can show that $(\veebar \mathcal{G})\veebar a^{\perp}=A$, for a finite subset $\mathcal{G}$ of $\mathcal{F}$. This results that $(\veebar_{G\in \mathcal{G}}\rho(G))\veebar a^{\perp}=A$. Therefore, $a\in \sigma(\veebar_{G\in \mathcal{G}}\rho(G))\subseteq \veebar_{F\in \mathcal{F}}\rho(F)$.
\item [\ref{equgelchapure6}$\Rightarrow$\ref{equgelchapure1}:] There is nothing to prove.
\item [\ref{equgelchapure7}$\Rightarrow$\ref{equgelchapure1}:] It follows by Theorem \ref{pmprop}.
\end{proof}
%%%%%%%%%%%%%%%%%%%%%%%%%%%%%%%%%%%%%%%%%%%%%%%%%%%%%%%%%%%%%%%%%%%%%%%%%%%%%%%%%%%%%%%%%%%%%%%%%%%%%%%%%%%%%%%%
Recall that a pair $(f,g)$ is called \textit{an adjunction (or isotone Galois connection)} between posets $\mathscr{A}=(A;\leq)$ and $\mathscr{B}=(B;\preccurlyeq)$, where $f:A\longrightarrow B$ and $g:B\longrightarrow A$ are two functions such that for all $a\in A$ and $b\in B$, $f(a)\leq b$ if and only if $a\preccurlyeq g(b)$. It is well-known that $(f,g)$ is an adjunction if and only if $gf$ is inflationary, $fg$ is deflationary, and $f,g$ are isotone \citep[Theorem 2]{garcia2013galois}. Also, $\mathscr{C}_{fg}=Im(f)$, in which $\mathscr{C}_{fg}$ is the set of fixed points of the kernel operator $fg$. In the following, the relationship between the pure filters and radicals in a Gelfand residuated lattice is described. This description is inspired by the one obtained for Gelfand rings \cite[\S 8, Proposition 34]{borceux1983algebra}.
\begin{theorem}\label{equgelchapure}
  Let $\mathfrak{A}$ be a residuated lattice. The following assertions are equivalent:
\begin{enumerate}
\item  [(1) \namedlabel{rhoradgel1}{(1)}] $\mathfrak{A}$ is Gelfand;
\item  [(2) \namedlabel{rhoradgel2}{(2)}] The pair $(\rho,Rad)$ is an adjunction.
\end{enumerate}
\end{theorem}
\begin{proof}
\item [\ref{rhoradgel1}$\Rightarrow$\ref{rhoradgel2}:] Let $F$ and $G$ be two filters of $\mathfrak{A}$. Suppose first $\rho(F)\subseteq G$. Assume that $\mathfrak{m}$ is a maximal filter of $\mathfrak{A}$ containing $G$. By Theorem \ref{equgelchapure} we have $F\subseteq \mathfrak{m}$, and this yields that $F\subseteq Rad(G)$. Conversely, suppose $F\subseteq Rad(G)$. One can see that $\rho(F)\subseteq \rho(\mathfrak{m})$, for any maximal filter $\mathfrak{m}$ of $\mathfrak{A}$ containing $G$. Using Theorem \ref{equgelchapure} and Proposition \ref{rfilter}\ref{rfilter4}, we have the following sequence of formulas:
  \[
  \begin{array}{ll}
    \rho(F) & \subseteq \{\rho(\mathfrak{m})\mid \mathfrak{m}\in h_{M}(G)\} \\
     & =\{\rho(\mathfrak{m})\mid \mathfrak{m}\in h_{M}(\rho(G))\} \\
     & =\rho(G)\subseteq G.
  \end{array}
  \]
\item [\ref{rhoradgel2}$\Rightarrow$\ref{rhoradgel1}:] It follows by Theorem \ref{equgelchapure}.
\end{proof}
%%%%%%%%%%%%%%%%%%%%%%%%%%%%%%%%%%%%%%%%%%%%%%%%%%%%%%%%%%%%%%%%%%%%%%%%%%%%%%%%%%%%%%%%%%%%%%%%%%%%%%%%%%%%%%%%
The subsequent theorem characterizes the purely-maximal filters of a Gelfand residuated lattice.
\begin{theorem}\label{gelfmaxpure}
  Let $\mathfrak{A}$ be a Gelfand residuated lattice. The following assertions hold:
\begin{enumerate}
\item  [$(1)$ \namedlabel{gelfmaxpure1}{$(1)$}] $Max(\sigma(\mathfrak{A}))=\rho(Max(\mathfrak{A}))$;
\item  [$(2)$ \namedlabel{gelfmaxpure2}{$(2)$}] $Spp(\mathfrak{A})=Max(\sigma(\mathfrak{A}))$;
\item  [$(3)$ \namedlabel{gelfmaxpure3}{$(3)$}] $Spp(\mathfrak{A})=\rho(Max(\mathfrak{A}))$.
\end{enumerate}
\end{theorem}
\begin{proof}
\item [\ref{gelfmaxpure1}:] It is a direct result of Proposition \ref{r1filter}\ref{r1filter2} and Theorem \ref{equgelchapure}.
\item [\ref{gelfmaxpure2}:] Let $\mathfrak{p}$ be a purely-prime filter of $\mathfrak{A}$. Assume that $\mathfrak{m}$ is a purely-maximal filter of $\mathfrak{A}$ containing $\mathfrak{p}$. Suppose that $\mathfrak{p}\neq \mathfrak{m}$. By Theorem \ref{equgelchapure}, we have $\mathfrak{m}=\veebar_{a\in \mathfrak{m}}\rho(\mathscr{F}(a))$. Choose $a\in \mathfrak{m}$ such that $\mathscr{F}(a)\nsubseteq \mathfrak{p}$. On the other hand, $\rho(\mathscr{F}(a)\cap \rho(a^{\perp})=\{1\}\subseteq \mathfrak{p}$ which implies that $\rho(a^{\perp})\subseteq \mathfrak{p}$; a contradiction.
\item [\ref{gelfmaxpure3}:] It follows by \ref{gelfmaxpure1} and \ref{gelfmaxpure2}.
\end{proof}
%%%%%%%%%%%%%%%%%%%%%%%%%%%%%%%%%%%%%%%%%%%%%%%%%%%%%%%%%%%%%%%%%%%%%%%%%%%%%%%%%%%%%%%%
\citet[\S 8, Theorems 39]{borceux1983algebra} has shown that the pure spectrum of a ring is Hausdorff. In the following theorem, this result has been generalized to residuated lattices.
\begin{proposition}\label{gelspphau}
  Let $\mathfrak{A}$ be a Gelfand residuated lattice. $Spp(\mathfrak{A})$ is a Hausdorff space.
\end{proposition}
\begin{proof}
Let $\mathfrak{p}$ and $\mathfrak{q}$ be two distinct purely-prime filters of $\mathfrak{A}$. Applying Theorem \ref{gelfmaxpure}\ref{gelfmaxpure3}, there exists two distinct maximal filters $\mathfrak{m}$ and $\mathfrak{n}$ such that $\mathfrak{p}=\rho(\mathfrak{m})$ and $\mathfrak{q}=\rho(\mathfrak{n})$. Using Theorem \ref{pmprop}, there exist $a\notin \mathfrak{m}$ and $b\notin \mathfrak{n}$ such that $a\vee b=1$. By Theorem \ref{equgelchapure} follows that $\rho(\mathscr{F}(a))\veebar \mathfrak{m}=A$ and $\rho(\mathscr{F}(b))\veebar \mathfrak{n}=A$. Thus $\mathfrak{m}\in d_{\kappa}(\rho(\mathscr{F}(a)))$ and $\mathfrak{n}\in d_{\kappa}(\rho(\mathscr{F}(b)))$. By Proposition \ref{rfilter}\ref{rfilter3}, we have the following sequence of formulas:
\[
\begin{array}{ll}
  d_{\kappa}(\rho(\mathscr{F}(a)))\cap d_{\kappa}(\rho(\mathscr{F}(b))) & =d_{\kappa}(\rho(\mathscr{F}(a))\cap \rho(\mathscr{F}(b))) \\
   & =d_{\kappa}(\rho(\mathscr{F}(a)\cap \mathscr{F}(b))) \\
   & =d_{\kappa}(\rho(\mathscr{F}(a\vee b)))=d_{\kappa}(\{1\})=\emptyset.
\end{array}
\]
This shows that the pure spectrum of $\mathfrak{A}$ is Hausdorff.
\end{proof}
%%%%%%%%%%%%%%%%%%%%%%%%%%%%%%%%%%%%%%%%%%%%%%%%%%%%%%%%%%%%%%%%%%%%%%%%%%%%%%%%%%%%%%%%
\citet[\S 8, Proposition 40]{borceux1983algebra} has shown that the pure spectrum of a Gelfand ring is homeomorphic to its usual maximal spectrum. In the following, this result is generalized and improved to Gelfand residuated lattices.
\begin{theorem}\label{sppgelfch}
  Let $\mathfrak{A}$ be a residuated lattice. The following assertions are equivalent:
\begin{enumerate}
\item  [(1) \namedlabel{sppgelfch1}{(1)}] $\mathfrak{A}$ is Gelfand;
\item  [(2) \namedlabel{sppgelfch2}{(2)}] the map $\rho_{m}:Max_{h}(\mathfrak{A})\longrightarrow Spp(\mathfrak{A})$, given by $\mathfrak{m}\rightsquigarrow \rho(\mathfrak{m})$, is a homeomorphism.
\end{enumerate}
\end{theorem}
\begin{proof}
\item [\ref{sppgelfch1}$\Rightarrow$\ref{sppgelfch2}:] Applying Theorem \ref{spsppconti}, $\rho_{m}$ is a well-defined continuous map. Injectivity of $\rho_{m}$ follows by Theorem \ref{equgelchapure}, and surjectivity of $\rho_{m}$ follows by Theorem \ref{gelfmaxpure}. Using Propositions \ref{hulkerinstr} and \ref{gelspphau}, it follows that $Max(\mathfrak{A})$ is compact, and $Spp(\mathfrak{A})$ is Hausdorff, respectively. This holds the result due to \citet[Theorem 3.1.13]{engelking1989general}.
\item [\ref{sppgelfch2}$\Rightarrow$\ref{sppgelfch1}:] One can see that $\rho_{m}^{-1}\circ \rho$ is a retraction from $Spec_{h}(\mathfrak{A})$ into $Max_{h}(\mathfrak{A})$. So $\mathfrak{A}$ is Gelfand due to Theorem \ref{gelnor}.
\end{proof}
%%%%%%%%%%%%%%%%%%%%%%%%%%%%%%%%%%%%%%%%%%%%%%%%%%%%%%%%%%%%%%%%%%%%%%%%%%%%%%%%%%%%
The pure ideals of a commutative reduced Gelfand ring with unity are characterized in \citet[Theorems 1.8 \& 1.9]{al1989pure}. These results have been improved and generalized to residuated lattices in Proposition \ref{gelpurefcl}.
\begin{theorem}\label{gelpurefcl}
Let $\mathfrak{A}$ be a Gelfand residuated lattice. The pure filters of $\mathfrak{A}$ are precisely of the form $\bigcap\{D(\mathfrak{m})\mid \mathfrak{m}\in Max(\mathfrak{A})\cap \mathcal{C}\}$, in which $\mathcal{C}$ is a closed subset of $Spec_{h}(\mathfrak{A})$.
\end{theorem}
\begin{proof}
let $a\in G:=\bigcap\{D(\mathfrak{m})\mid \mathfrak{m}\in Max(\mathfrak{A})\cap \mathcal{C}\}$, in which $\mathcal{C}$ is a closed subset of $Spec_{h}(\mathfrak{A})$. So for any $\mathfrak{m}\in Max(\mathfrak{A})\cap \mathcal{C}$, we have $\mathfrak{m}\veebar a^{\perp}=A$. By absurdum, assume that $G\veebar a^{\perp}\neq A$. So $G\veebar a^{\perp}$ is contained in a maximal filter $\mathfrak{n}$. Obviously, $\mathfrak{n}\notin \mathcal{C}$. So for any $\mathfrak{m}\in  Max(\mathfrak{A})\cap \mathcal{C}$, there exist $x_{\mathfrak{m}}\notin \mathfrak{m}$ and $y_{\mathfrak{m}}\notin \mathfrak{n}$ such that $x_{\mathfrak{m}}\vee y_{\mathfrak{m}}=1$. Since $\mathcal{C}$ is stable under the specialization, so $\mathcal{C}\subseteq \bigcup_{\mathfrak{m}\in Max(\mathfrak{A})\cap \mathcal{C}} d(x_{\mathfrak{m}})$. By Proposition \ref{hulkerinstr} follows that $\mathcal{C}$ is compact. So there exist a finite number $\mathfrak{m}_{1},\cdots,\mathfrak{m}_{n}\in Max(\mathfrak{A})\cap \mathcal{C}$ such that $\mathcal{C}\subseteq \bigcup_{i=1}^{n} d(x_{\mathfrak{m}_{i}})$. Set $x=\bigodot_{i=1}^{n} x_{\mathfrak{m}_{i}}$ and $y=\bigvee_{i=1}^{n}y_{\mathfrak{m}_{i}}$. Routinely, one can see that $x\in y^{\perp}\setminus \mathfrak{m}$, for any $\mathfrak{m}\in Max(\mathfrak{A})\cap \mathcal{C}$. This implies that $y\in G$; a contradiction. The converse follows by Theorem \ref{sigmafequiv}\ref{sigmafequiv5}.
\end{proof}
%%%%%%%%%%%%%%%%%%%%%%%%%%%%%%%%%%%%%%%%%%%%%%%%%%%%%%%%%%%%%%%%%%%%%%%%%%%%%%%%%%%%
The subsequent theorem gives some criteria for a residuated lattice to be Gelfand, inspired by the one obtained for bounded distributive lattices by \citet[Theorem 6]{al1990topological}.
\begin{theorem}\label{gelfhulldmin}
  Let $\mathfrak{A}$ be a residuated lattice. The following assertions are equivalent:
\begin{enumerate}
\item  [(1) \namedlabel{gelfhulldmin1}{(1)}] $\mathfrak{A}$ is Gelfand;
\item  [(2) \namedlabel{gelfhulldmin2}{(2)}] the hull-kernel and $\mathscr{D}$-topology coincide on $Max(\mathfrak{A})$.
\end{enumerate}
\end{theorem}
\begin{proof}
\item [\ref{gelfhulldmin1}$\Rightarrow$\ref{gelfhulldmin2}:] Let $f$ be a retraction from $Spec(\mathfrak{A})$ into $Max(\mathfrak{A})$. Consider $a\in A$. Routinely, one can show that $d_{M}(a)=f^{\leftarrow}(d_{M}(a))\cap Max(\mathfrak{A})$, in which $f^{\leftarrow}(d_{M}(a))$ is open $\mathscr{S}$-stable. This yields that $\tau_{\mathscr{D}}$ is finer than $\tau_{h}$ on $Max(\mathfrak{A})$.
\item [\ref{gelfhulldmin2}$\Rightarrow$\ref{gelfhulldmin1}:] Let $\mathfrak{p}$ be a prime filter of $\mathfrak{A}$ that is contained in two distinct maximal filters $\mathfrak{m}$ and $\mathfrak{n}$. So there exists $a\in A$ such that $\mathfrak{m}\in d(a)$ and $\mathfrak{n}\notin d(a)$. Thus $d_{M}(a)=d_{M}(F)$, for some pure filter $F$ of $\mathfrak{A}$. Since $\mathfrak{m}\in d(F)$, so $\mathfrak{p}\in d(F)$. This implies that $\mathfrak{n}\in d(F)$; a contradiction.
\end{proof}
%%%%%%%%%%%%%%%%%%%%%%%%%%%%%%%%%%%%%%%%%%%%%%%%%%%%%%%%%%%%%%%%%%%%%%%%%%%%%%%%%%%%%%%%%
%%%%%%%%%%%%%%%%%%%%%%%%%%%%%%%%%%%%%%%%%%%%%%%%%%%%%%%%%%%%%%%%%%%%%%%%%%%%%%%%%%%%%%%%%
\section{The pure spectrum of an mp-residuated lattice}\label{sec6}
This section deals with the pure spectrum of an mp-residuated lattice. For the basic facts concerning mp-residuated lattices, the interesting reader is referred to \cite{rasouli2022mp}.
\begin{definition}\label{nordef}
A residuated lattice $\mathfrak{A}$ is called \textit{mp} provided that any prime filter of $\mathfrak{A}$ contains a unique minimal prime filter of $\mathfrak{A}$.
\end{definition}
\begin{remark}
  By the above definition, mp-residuated lattices can be regarded even as the dual notion of Gelfand residuated lattices.
\end{remark}
\begin{example}\label{quanorempxas}
One can see that the residuated lattice $\mathfrak{A}_6$ from Example \ref{exa6},the residuated lattice $\mathfrak{B}_6$ from Example \ref{exb6}, and the residuated lattice $\mathfrak{C}_6$ from Example \ref{exc6} are mp, and the residuated lattice $\mathfrak{A}_8$ from Example \ref{exa8} is not mp.
\end{example}
\begin{example}
   The class of MTL-algebras, and so,  MV-algebras, BL-algebras, and Boolean algebras are some subclasses of mp-residuated lattices.
\end{example}

The following proposition gives some algebraic characterizations for mp-residuated lattices.
\begin{proposition}\cite[Theorem 3.5]{rasouli2022mp}\label{noco}
Let $\mathfrak{A}$ be a residuated lattice. The following assertions are equivalent:
\begin{enumerate}
\item  [$(1)$ \namedlabel{noco1}{$(1)$}] Any two distinct minimal prime filters are comaximal;
\item  [$(2)$ \namedlabel{noco2}{$(2)$}] $\mathfrak{A}$ is mp;
%\item  [$(3)$ \namedlabel{noco3}{$(3)$}] for any prime filter $\mathfrak{p}$ of $\mathfrak{A}$, $D(\mathfrak{p})$ is a prime filter of $\mathfrak{A}$;
\item  [$(4)$ \namedlabel{noco4}{$(4)$}] for any maximal filter $\mathfrak{m}$ of $\mathfrak{A}$, $D(\mathfrak{m})$ is a minimal prime filter of $\mathfrak{A}$;
\item  [$(5)$ \namedlabel{noco5}{$(5)$}] for any pairwise elements $x$ and $y$ in $A$, $x^{\perp}\veebar y^{\perp}=A$;
%\item  [$(6)$ \namedlabel{noco6}{$(6)$}] for any pairwise elements $x$ and $y$ in $A$, there exists $a\in A$ such that $a\in x^{\perp}$ and $\neg a\in y^{\perp}$;
%\item  [$(7)$ \namedlabel{noco7}{$(7)$}] for any $x,y\in A$, $(x\vee y)^{\perp}=x^{\perp}\veebar y^{\perp}$;
%\item  [$(8)$ \namedlabel{noco8}{$(8)$}] for any $x,y\in A$, $(x\vee y)^{\perp}=A$ implies $x^{\perp}\veebar y^{\perp}=A$.
\end{enumerate}
\end{proposition}

The subsequent theorem, which can be compared with Theorem \ref{gelnor}, gives some necessary and sufficient conditions for the collection of minimal prime filters in a residuated lattice to be a Hausdorff space with the dual hull-kernel topology.
\begin{theorem}\label{mpmpropd}
Let $\mathfrak{A}$ be a residuated lattice.  The following assertions are equivalent:
 \begin{enumerate}
   \item [(1) \namedlabel{mpmpropd1}{(1)}] $\mathfrak{A}$ is mp;
   \item [(2) \namedlabel{mpmpropd2}{(2)}] $Min_{d}(\mathfrak{A})$ is Hausdorff.
 \end{enumerate}
\end{theorem}

The following theorem gives some criteria for mp-residuated lattices by pure filters, inspired by the
one obtained for bounded distributive lattices by \citet[Theorem 2.11]{cornish1977ideals}.
\begin{theorem}\label{norgammsig}
  Let $\mathfrak{A}$ be a residuated lattice. The following assertions are equivalent:
  \begin{enumerate}
\item  [$(1)$ \namedlabel{norgammsig1}{$(1)$}] $\mathfrak{A}$ is mp;
\item  [$(2)$ \namedlabel{norgammsig2}{$(2)$}] $\Omega(\mathfrak{A})\subseteq \sigma(\mathfrak{A})$;
\item  [$(3)$ \namedlabel{norgammsig3}{$(3)$}] $\gamma(\mathfrak{A})\subseteq \sigma(\mathfrak{A})$.
\end{enumerate}
\end{theorem}
\begin{proof}
\item [\ref{norgammsig1}$\Rightarrow$\ref{norgammsig2}:] Let $F$ be an $\omega$-filter of $\mathfrak{A}$. So $F=\omega(I)$, for some ideal $I$ of $\ell(\mathfrak{A})$. Consider $x\in F$. So $x\in a^{\perp}$, for some $a\in I$. By Propositions \ref{omegprop}\ref{omegprop1} and \ref{noco}\ref{noco4} follows that $A=x^{\perp}\veebar a^{\perp}\subseteq x^{\perp}\veebar F$.
\item [\ref{norgammsig2}$\Rightarrow$\ref{norgammsig3}:] By Propositions \ref{omegprop}\ref{omegprop2}, it is evident.
\item [\ref{norgammsig3}$\Rightarrow$\ref{norgammsig1}:] Let $x\vee y=1$. So $x\in y^{\perp}=\sigma(y^{\perp})$ and this implies that $x^{\perp}\veebar y^{\perp}=A$. Hence the result holds by Proposition \ref{noco}\ref{noco5}.
\end{proof}
\begin{remark}
  \citet[Theorem 1]{al1987some} showed that a unitary commutative ring is a PF ring if and only if any its annulet is a pure ideal. Thus, if we define PF-residuated lattices as those ones in which any coannulet is a pure filter, Theorem \ref{norgammsig} verifies that the class of PF residuated lattices coincides with the class of mp-residuated lattices.
\end{remark}
%%%%%%%%%%%%%%%%%%%%%%%%%%%%%%%%%%%%%%%%%%%%%%%%%%%%%%%%%%%%%%%%%%%%%%%%%%%%%%%%%%%%%%%%%
\begin{theorem}\label{norgammsige}
  Let $\mathfrak{A}$ be a residuated lattice. The following assertions are equivalent:
  \begin{enumerate}
\item  [$(1)$ \namedlabel{norgammsige1}{$(1)$}] $\mathfrak{A}$ is mp;
\item  [$(2)$ \namedlabel{norgammsige2}{$(2)$}] $D(\mathfrak{p})$ is a pure filter of $\mathfrak{A}$, for any prime filter $\mathfrak{p}$ of $\mathfrak{A}$;
\item  [$(3)$ \namedlabel{norgammsige3}{$(3)$}] $D(\mathfrak{m})$ is a pure filter of $\mathfrak{A}$, for any maximal filter $\mathfrak{m}$ of $\mathfrak{A}$;
\item  [$(4)$ \namedlabel{norgammsige4}{$(4)$}] $Min(\mathfrak{A})\subseteq \sigma(\mathfrak{A})$.
\end{enumerate}
\end{theorem}
\begin{proof}
\item [\ref{norgammsige1}$\Rightarrow$\ref{norgammsige2}:] It follows by Theorem \ref{norgammsig}.
\item [\ref{norgammsige2}$\Rightarrow$\ref{norgammsige3}:] It is evident.
\item [\ref{norgammsige3}$\Rightarrow$\ref{norgammsige4}:] It follows, with a little bit of effort, by Proposition \ref{omegprop}\ref{omegprop3}.
\item [\ref{norgammsige4}$\Rightarrow$\ref{norgammsige1}:] It follows by Propositions \ref{comxpureprime} \& \ref{noco}\ref{noco1}.
\end{proof}
%%%%%%%%%%%%%%%%%%%%%%%%%%%%%%%%%%%%%%%%%%%%%%%%%%%%%%%%%%%%%%%%%%%%%%%%%%%%%%%%%%%%%%%%%
\begin{theorem}\label{normpurprimxa}
  Let $\mathfrak{A}$ be a residuated lattice. The following assertions are equivalent:
\begin{enumerate}
\item  [(1) \namedlabel{normpurprimxa1}{(1)}] $\mathfrak{A}$ is mp;
\item  [(2) \namedlabel{normpurprimxa2}{(2)}] $Min(\mathfrak{A})=Max(\sigma(\mathfrak{A}))$.
\end{enumerate}
\end{theorem}
\begin{proof}
\item [\ref{normpurprimxa1}$\Rightarrow$\ref{normpurprimxa2}:] Let $\mathfrak{m}$ be a minimal prime filter of $\mathfrak{A}$. By Theorem \ref{norgammsige} follows that $\mathfrak{m}$ is a pure filter of $\mathfrak{A}$. Thus there exists $\mathfrak{n}\in Max(\sigma(\mathfrak{A}))$ containing $\mathfrak{m}$. Let $a\in \mathfrak{n}$. So there exists $b\in a^{\perp}$ such that $\neg b\in \mathfrak{n}$. This implies that $b\notin \mathfrak{m}$, and so $a\in \mathfrak{m}$. Conversely, let $\mathfrak{p}$ be a purely-maximal filter of $\mathfrak{A}$. So $\mathfrak{p}\subseteq \mathfrak{n}$, for some $\mathfrak{n}\in Max(\mathfrak{A})$. Using Propositions \ref{sigmapro}(\ref{sigmapro2} \& \ref{sigmapro4}) \& \ref{noco}\ref{noco4}, and Theorem \ref{norgammsige}\ref{norgammsige3}, it shows that $\mathfrak{p}=D(\mathfrak{n})\in Min(\mathfrak{A})$.
\item [\ref{normpurprimxa2}$\Rightarrow$\ref{normpurprimxa1}:] It is evident by Theorem \ref{norgammsige}\ref{norgammsige4}.
\end{proof}
%%%%%%%%%%%%%%%%%%%%%%%%%%%%%%%%%%%%%%%%%%%%%%%%%%%%%%%%%%%%%%%%%%%%%%%%%%%%%%%%%%%%%%%%
The following result generalized and improved \cite[Theorem 1.8]{al1989pure} to residuated lattices.
\begin{proposition}\label{pureinterd}
  Let $\mathfrak{A}$ be an mp-residuated lattice and $F$ a proper pure filter of $\mathfrak{A}$. We have
  \[F=kh_{m}(F).\]
\end{proposition}
\begin{proof}
  By Theorem \ref{normpurprimxa}, $h_{m}(F)\neq \emptyset$. Consider $a\in kh_{m}(F)$. Assume that $a^{\perp}\veebar F$ is proper. Thus $a^{\perp}\veebar F\subseteq \mathfrak{n}$, for some maximal filter $\mathfrak{n}$ of $\mathfrak{A}$. Let $\mathfrak{m}$ be a minimal prime filter of $\mathfrak{A}$ contained in $\mathfrak{n}$. This implies that $F\subseteq \mathfrak{m}$, and so $\neg b\in \mathfrak{n}$, for some $b\in a^{\perp}$; a contradiction.
\end{proof}
%%%%%%%%%%%%%%%%%%%%%%%%%%%%%%%%%%%%%%%%%%%%%%%%%%%%%%%%%%%%%%%%%%%%%%%%%%%%%%%%%%%%%%%%
The pure ideals of a PF ring are characterized in \citet[Theorems 2.4 and 2.5]{al1988pure}. These results have been improved and generalized to residuated lattices in Theorem \ref{mppurefcl} and Proposition \ref{mppure}. The following result can be compared with Theorem \ref{gelpurefcl}.
\begin{theorem}\label{mppurefcl}
Let $\mathfrak{A}$ be an mp-residuated lattice. The pure filters of $\mathfrak{A}$ are precisely of the form $\bigcap_{\mathfrak{m}\in Min(\mathfrak{A})\cap \mathcal{C}}\mathfrak{m}$,where $\mathcal{C}$ runs over closed subsets of $Spec_{d}(\mathfrak{A})$.
\end{theorem}
\begin{proof}
let $a\in G:=\bigcap\{\mathfrak{m}\mid \mathfrak{m}\in Min(\mathfrak{A})\cap \mathcal{C}\}$, in which $\mathcal{C}$ is a closed subset of $Spec_{d}(\mathfrak{A})$. So for any $\mathfrak{m}\in Min(\mathfrak{A})\cap \mathcal{C}$, we have $\mathfrak{m}\veebar a^{\perp}=A$. By absurdum, assume that $G\veebar a^{\perp}\neq A$. So $G\veebar a^{\perp}$ is contained in a maximal filter $\mathfrak{n}$. Let $\mathfrak{o}$ be a minimal prime filter of $\mathfrak{A}$ contained in $\mathfrak{m}$. Obviously, $\mathfrak{o}\notin \mathcal{C}$. So for any $\mathfrak{m}\in  Min(\mathfrak{A})\cap \mathcal{C}$, there exist $x_{\mathfrak{m}}\in \mathfrak{m}$ and $y_{\mathfrak{m}}\in \mathfrak{o}$ such that $x_{\mathfrak{m}}\odot y_{\mathfrak{m}}=0$. Since $\mathcal{C}$ is stable under the generalization, so $\mathcal{C}\subseteq \bigcup_{\mathfrak{m}\in Min(\mathfrak{A})\cap \mathcal{C}} h(x_{\mathfrak{m}})$. By Proposition \ref{hulkerinstr} follows that $\mathcal{C}$ is compact. So there exist $\mathfrak{m}_{1},\cdots,\mathfrak{m}_{n}\in Min(\mathfrak{A})\cap \mathcal{C}$ such that $\mathcal{C}\subseteq \bigcup_{i=1}^{n} h(x_{\mathfrak{m}_{i}})$. Set $x=\bigvee_{i=1}^{n} x_{\mathfrak{m}_{i}}$ and $y=\bigodot_{i=1}^{n}y_{\mathfrak{m}_{i}}$. Routinely, one can see that $0=x\odot y\in G\veebar \mathfrak{o}$; a contradiction. The converse follows by Proposition \ref{pureinterd}.
\end{proof}
%%%%%%%%%%%%%%%%%%%%%%%%%%%%%%%%%%%%%%%%%%%%%%%%%%%%%%%%%%%%%%%%%%%%%%%%%%%%%%%%%%%%%%%%%
\begin{proposition}\label{mppureco1}
Let $\mathfrak{A}$ be an mp-residuated lattice and $a\in A$. Then
\begin{center}
  $a^{\perp}\cap F_{a}=\{1\}$, where $F_{a}=\bigcap_{\mathfrak{m}\in h_{M}(a)}\rho(\mathfrak{m})$.
\end{center}
\end{proposition}
\begin{proof}
With a little bit of effort, it follows by Theorem \ref{norgammsig} and Proposition \ref{rfilter}.
\end{proof}
\begin{corollary}\label{mppu1re}
  If $\mathfrak{m}$ is a minimal prime filter of an mp-residuated lattice $\mathfrak{A}$, then $\mathfrak{m}=\underline{\bigvee}_{a\in \mathfrak{m}}F_{a}$.
\end{corollary}
\begin{proof}
Let $a\in \mathfrak{m}$. So $b\in a^{\perp}$, for some $b\notin \mathfrak{m}$. This implies that $a\in F_{\neg b}$. The reverse inclusion is deduced from Corollary \ref{mppureco1}.
\end{proof}
%%%%%%%%%%%%%%%%%%%%%%%%%%%%%%%%%%%%%%%%%%%%%%%%%%%%%%%%%%%%%%%%%%%%%%%%%%%%%%%%%%%%%%%%%%%%%%%%%%%%%%%%%%%%%%%%%%%%%%%%%%%%%
\begin{proposition}\label{mppure}
Let $\mathfrak{A}$ be an mp-residuated lattice. The pure filters of $\mathfrak{A}$ are precisely of the form $\bigcap_{\mathfrak{m}\in h_{M}(F)}\rho(\mathfrak{m})$, where $F$ is a filter of $\mathfrak{A}$.
\end{proposition}
\begin{proof}
Let $\mathcal{C}=\{P\in Spec(\mathfrak{A})\mid P\cap \neg F=\emptyset\}$. One can see that $h_{M}(F)=Min(\mathfrak{A})\cap \mathcal{C}$. This establishes the result due to \textsc{Remark} \ref{closd} and Theorem \ref{mppurefcl}.
\end{proof}

\citet[Theorem 3.5]{al1988pure} proved that every purely prime ideal of a PF ring is purely maximal. Now we provide an alternative proof to the following interesting result.
\begin{theorem}\label{mpminspp}
  Let $\mathfrak{A}$ be an mp residuated lattice. Then
  \[Spp(\mathfrak{A})\subseteq Max(\sigma(\mathfrak{A})).\]
\end{theorem}
\begin{proof}
Let $\mathfrak{p}$ be a purely prime filter of $\mathfrak{A}$. So $\mathfrak{p}\subseteq \mathfrak{m}$, for some $\mathfrak{m}\in Max(\sigma(\mathfrak{A}))$. By Theorem \ref{normpurprimxa} we have $\mathfrak{m}\in Min(\mathfrak{A})$. Let $a\in \mathfrak{m}$. By Proposition \ref{1mineq} we have $a^{\perp}\nsubseteq \mathfrak{m}$. By Proposition \ref{mppureco1} follows that $a^{\perp}\cap F_{a}\subseteq P$. By Theorem \ref{norgammsig} and Proposition \ref{mppure}, respectively, follows that $a^{\perp}$ and $F_{a}$ are pure filters. This implies that $F_{a}\subseteq \mathfrak{p}$. Hence by Corollary \ref{mppu1re} follows that  $\mathfrak{m}=\underline{\bigvee}_{a\in \mathfrak{m}}F_{a}\subseteq \mathfrak{p}$.
\end{proof}
\begin{remark}
 Lemma \ref{prinpuregen} and Theorems \ref{gelfmaxpure}\ref{gelfmaxpure2} \& \ref{mpminspp} show that in finite, Gelfand, and mp residuated lattices, respectively, every purely-prime filter is purely-maximal. So it seems that finding a residuated lattice with a purely-prime filter that is not a purely-maximal filter is not easy. Hence, arises a question: Is every purely-prime filter of a residuated lattice purely-maximal? It seems to us that the yes or no answer to this question is not simple. Although, we guess that the negative answer to this conjecture looks more likely.
\end{remark}

The main result of this section is the following theorem which establishes that a residuated lattice is mp if and only if the set of its minimal prime filters and the set of its purely-prime filters coincide.
\begin{theorem}\label{mp2minspp}
  Let $\mathfrak{A}$ be a residuated lattice. The following assertions are equivalent:
\begin{enumerate}
\item  [(1) \namedlabel{mp2minspp1}{(1)}] $\mathfrak{A}$ is mp;
\item  [(2) \namedlabel{mp2minspp2}{(2)}] $Min(\mathfrak{A})=Spp(\mathfrak{A})$.
\end{enumerate}
\end{theorem}
\begin{proof}
\item [\ref{mp2minspp1}$\Rightarrow$\ref{mp2minspp2}:] It follows by Theorems \ref{normpurprimxa} \& \ref{mpminspp}.
\item [\ref{mp2minspp2}$\Rightarrow$\ref{mp2minspp1}:] It follows by Theorem \ref{norgammsige}.
\end{proof}
%%%%%%%%%%%%%%%%%%%%%%%%%%%%%%%%%%%%%%%%%%%%%%%%%%%%%%%%%%%%%%%%%%%%%%%%%%%%%%%%%%%%%%%%%%%%%%%%%%%%%%%%%%%%%%%%%
The subsequent result, which can be compared with Theorem \ref{sppgelfch}, gives a criterion for a residuated lattice to be mp
\begin{theorem}\label{equmpflatmin}
  Let $\mathfrak{A}$ be a residuated lattice. The following assertions are equivalent:
\begin{enumerate}
\item  [(1) \namedlabel{equmpflatmin1}{(1)}] $\mathfrak{A}$ is mp;
\item  [(2) \namedlabel{equmpflatmin2}{(2)}] the identity map $\iota:Spp(\mathfrak{A})\longrightarrow Min_{d}(\mathfrak{A})$ is a homeomorphism.
\end{enumerate}
\end{theorem}
\begin{proof}
\item [\ref{equmpflatmin1}$\Rightarrow$\ref{equmpflatmin2}:] Consider the identity map $\iota:Spp(\mathfrak{A})\longrightarrow Min(\mathfrak{A})$. Using Theorem \ref{mp2minspp}, it follows that $\iota$ is a well-defined bijection. One can see that $h_{m}(a)=d_{\kappa}(a^{\perp})$, for any $a\in A$, which implies that $\iota$ is continuous. By Theorems \ref{puresppcomp} \& \ref{mpmpropd}, it follows that $Min_{d}(\mathfrak{A})$ is Hausdorff, and $Spp(\mathfrak{A})$ is compact, respectively. This holds the result due to \citet[Theorem 3.1.13]{engelking1989general}.
\item [\ref{equmpflatmin2}$\Rightarrow$\ref{equmpflatmin1}:] It is evident by Theorem \ref{mp2minspp}.
\end{proof}
%%%%%%%%%%%%%%%%%%%%%%%%%%%%%%%%%%%%%%%%%%%%%%%%%%%%%%%%%%%%%%%%%%%%%%%%%%%%%%%%%%%%%%%%%
The following result, which can be compared with Proposition \ref{gelspphau}, verifies that the pure spectrum of an mp-residuated lattice is Hausdorff.
\begin{corollary}\label{mpspphau}
  Let $\mathfrak{A}$ be an mp-residuated lattice. $Spp(\mathfrak{A})$ is a Hausdorff space.
\end{corollary}
\begin{proof}
It is an immediate consequence of Theorems \ref{mpmpropd} \& \ref{equmpflatmin}.
\end{proof}

Notice that if $\mathfrak{A}$ is an mp-residuated lattice, by Theorem \ref{mp2minspp}, $Min(\mathfrak{A})=Spp(\mathfrak{A})$ as set, but not as topological space. Indeed, the induced hull-kernel topology over $Min(\mathfrak{A})$ is finer than the pure topology. The subsequent theorem, which is inspired by the one obtained for conormal lattices by \citet[Theorem 3.6]{al1991sigma}, establishes that the equality holds if and only if $Min_{h}(\mathfrak{A})$ is compact.
\begin{theorem}\label{minspprick}
  Let $\mathfrak{A}$ be an mp-residuated lattice. The following assertions are equivalent:
\begin{enumerate}
\item  [(1) \namedlabel{equmpflatmin1}{(1)}] $Min_{h}(\mathfrak{A})$ is compact;
\item  [(2) \namedlabel{equmpflatmin2}{(2)}] $Min_{h}(\mathfrak{A})$ and $Spp(\mathfrak{A})$ are homeomorphic.
\end{enumerate}
\end{theorem}
\begin{proof}
\item [\ref{equmpflatmin1}$\Rightarrow$\ref{equmpflatmin2}:] Consider the identity map $\iota:Min(\mathfrak{A})\longrightarrow Spp(\mathfrak{A})$. Using Theorem \ref{mp2minspp}, it follows that $\iota$ is a well-defined continuous bijection. Since $Min_{h}(\mathfrak{A})$ is compact and $Spp(\mathfrak{A})$ is Hausdorff the result holds by \citet[Theorem 3.1.13]{engelking1989general}.
\item [\ref{equmpflatmin2}$\Rightarrow$\ref{equmpflatmin1}:] It is evident by Theorem \ref{puresppcomp}.
\end{proof}

%%%%%%%%%%%%%%%%%%%%%%%%%%%%%%%%%%%%%%%%%%%%%%%%%%%%%%%%%%%%%%%%%%%%%%%%%%%%%%%%%%%%%%%%%
%%%%%%%%%%%%%%%%%%%%%%%%%%%%%%%%%%%%%%%%%%%%%%%%%%%%%%%%%%%%%%%%%%%%%%%%%%%%%%%%%%%%%%%%%%%%%%%%%%%%%%%%%%%%%%%%%%%%%%

\end{document}